\documentclass[letterpaper,11pt]{article}
\usepackage[utf8]{inputenc}
\usepackage[english]{babel}

\usepackage{booktabs}
\usepackage{multirow}

\usepackage{xspace}
\usepackage[tbtags]{amsmath}
\usepackage{amsfonts,amsthm,amssymb}
\usepackage{mathtools}
\usepackage{thmtools}
\usepackage{bm}
\usepackage{verbatim}
\usepackage{algorithm}
\usepackage{algorithmicx}
\usepackage{subcaption}

\usepackage[margin=1in]{geometry}

\usepackage[dvipsnames,usenames]{xcolor}
\usepackage[colorlinks=true,pdfpagemode=UseNone,urlcolor=RoyalBlue,linkcolor=RoyalBlue,citecolor=OliveGreen,pdfstartview=FitH]{hyperref}

\usepackage{tcolorbox}

\usepackage{todonotes} 
\usepackage{tablefootnote}

\usepackage{mleftright}

\usepackage{dsfont}

\declaretheorem{definition}

\declaretheorem{theorem}
\declaretheorem[sibling=theorem]{lemma}
\declaretheorem[sibling=theorem]{claim}
\declaretheorem[sibling=theorem]{corollary}

\declaretheorem{remark}

\newcommand{\bin}{\{0,1\}}
\newcommand{\abs}[1]{\left| #1 \right|}
\newcommand{\pbra}[1]{\left( #1 \right)}
\newcommand{\cbra}[1]{\left\{ #1 \right\}}
\newcommand{\sbra}[1]{\left[ #1 \right]}

\newcommand{\supp}{\mathsf{supp}}
\def\calX{{\mathcal X}}

\newcommand{\Enc}{\mathsf{Enc}}
\newcommand{\Samp}{\mathsf{Samp}}
\newcommand{\srSamp}{\mathsf{SRSamp}}

\newcommand{\anmExt}{\mathsf{anmExt}}

\newcommand{\BAD}{\mathrm{BAD}}
\newcommand{\Maj}{\mathsf{Maj}}
\newcommand{\LSRExt}{\mathsf{LSRExt}}

\newcommand{\BFExt}{\mathsf{BFExt}}
\newcommand{\nmExt}{\mathsf{nmExt}}
\newcommand{\LExt}{\mathsf{LExt}}

\newcommand{\AffineAdvCB}{\mathsf{AffineAdvCB}}

\newcommand{\AffineAdvGen}{\mathsf{AffineAdvGen}}
\newcommand{\advGen}{\mathsf{advGen}}

\newcommand{\Slice}{\mathsf{Slice}}

\newcommand{\AExt}{\mathsf{AExt}}
\newcommand{\Ext}{\mathsf{Ext}}

\newcommand{\WROLBP}{\mathsf{WROLBP}}
\newcommand{\SROLBP}{\mathsf{SROLBP}}
\newcommand{\ROBP}{\mathsf{ROBP}}

\newcommand{\A}{\mathcal{A}}

\newcommand{\F}{\mathsf{F}}
\newcommand{\N}{\mathbb{N}}
\newcommand{\D}{\cal D}

\newcommand{\poly}{\mathsf{poly}}

\newcommand{\eps}{\varepsilon}
\newcommand{\ac}{\mathsf{AC}}
\newcommand{\E}{\mathbf{E}}
\renewcommand{\Pr}{\mathbf{Pr}}

\newcommand{\polylog}{\mathrm{polylog}}

\newcommand{\avgH}{\widetilde{H}_\infty}
\newcommand{\minH}{H_\infty}
\newcommand{\tnmext}{\mathsf{2nmExt}}

\newcommand{\I}{\mathbf{I}}

\newcommand{\Alpha}{\mathrm{A}}
\newcommand{\NOBF}{\mathrm{NOBF}}

\newcommand{\Supp}{\mathsf{Supp}}

\newcommand{\post}{\mathsf{Post}}
\newcommand{\pre}{\mathsf{Pre}}

\newcommand{\zo}{\{0, 1\}}
\newcommand{\bits}{\{0, 1\}}

\newcommand{\BI}{\begin{itemize}}
\newcommand{\EI}{\end{itemize}}
\newcommand{\BE}{\begin{enumerate}}
\newcommand{\EE}{\end{enumerate}}
\newcommand{\eqdef}{\vcentcolon=}

\newcommand{\BT}{\begin{theorem}}   \newcommand{\ET}{\end{theorem}}
\newcommand{\BD}{\begin{definition}}   \newcommand{\ED}{\end{definition}}
\newcommand{\BCR}{\begin{corollary}} \newcommand{\ECR}{\end{corollary}}
\newcommand{\BCT}{\begin{constr}} \newcommand{\ECT}{\end{constr}}
\newcommand{\BL}{\begin{lemma}}   \newcommand{\EL}{\end{lemma}}
\newcommand{\BCM}{\begin{claim}}   \newcommand{\ECM}{\end{claim}}

\title{Two-Source and Affine Non-Malleable Extractors for Small Entropy}

\author{
Xin Li \thanks{Department of Computer Science, Johns Hopkins University, \texttt{lixints@cs.jhu.edu}. Supported by NSF CAREER Award CCF-1845349 and NSF Award CCF-2127575.}
\and
Yan Zhong \thanks{Department of Computer Science, Johns Hopkins University, \texttt{yzhong36@jhu.edu}. Supported by NSF CAREER Award CCF-1845349.}
}

\begin{document}
\title{Two-Source and Affine Non-Malleable Extractors for Small Entropy}

\begin{titlepage}
   \thispagestyle{empty}
    \maketitle
    \begin{abstract}
        \thispagestyle{empty}
        Non-malleable extractors are generalizations and strengthening of standard randomness extractors, that are resilient to adversarial tampering. Such extractors have wide applications in cryptography and have become important cornerstones in recent breakthroughs of explicit constructions of two-source extractors and affine extractors for small entropy.\ However, explicit constructions of non-malleable extractors appear to be much harder than standard extractors.\ Indeed, in the well-studied models of two-source and affine non-malleable extractors, the previous best constructions only work for entropy rate $>2/3$ and $1-\gamma$ for some small constant $\gamma>0$ respectively by Li (FOCS' 23).

In this paper, we present explicit constructions of two-source and affine non-malleable extractors that match the state-of-the-art constructions of standard ones for small entropy. Our main results include:
\begin{itemize}
    \item Two-source and affine non-malleable extractors (over $\F_2$) for sources on $n$ bits with min-entropy $k \ge \log^C n$ and polynomially small error, matching the parameters of standard extractors by Chattopadhyay and Zuckerman (STOC' 16, Annals of Mathematics' 19) and Li (FOCS' 16).
    \item Two-source and affine non-malleable extractors (over $\F_2$) for sources on $n$ bits with min-entropy $k = O(\log n)$ and constant error, matching the parameters of standard extractors by Li (FOCS' 23).
\end{itemize}

Our constructions significantly improve previous results, and the parameters (entropy requirement and error) are the best possible without first improving the constructions of standard extractors. In addition, our improved affine non-malleable extractors give strong lower bounds for a certain kind of read-once linear branching programs, recently introduced by  Gryaznov, Pudl\'{a}k, and Talebanfard (CCC' 22) as a generalization of several well studied computational models. These bounds match the previously best-known average-case hardness results given by Chattopadhyay and Liao (CCC' 23) and Li (FOCS'23), where the branching program size lower bounds are close to optimal, but the explicit functions we use here are different.\ Our results also suggest a possible deeper connection between non-malleable extractors and standard ones.

    \end{abstract}
\end{titlepage}

\tableofcontents
\thispagestyle{empty}
\newpage
\addtocontents{toc}{\protect\thispagestyle{empty}} 
\setcounter{page}{1}

\numberwithin{equation}{section}
\numberwithin{theorem}{section}
\numberwithin{lemma}{section}
\numberwithin{definition}{section}
\numberwithin{corollary}{section}

\section{Introduction}
\label{sec:intro}

Randomness extractors are fundamental objects in the broad area of pseudorandomness. These objects have been studied extensively and found applications in diverse areas such as cryptography, complexity theory, combinatorics and graph theory, and many more. Informally, randomness extractors are functions that transform imperfect randomness called weak random sources into almost uniform random bits. Originally, the motivation for studying these objects comes from the gap between the requirement of high-quality random bits in various computational and cryptographic applications, and the severe biases in natural random sources. In practice, weak random sources can arise in several different situations. For example, the random bits can become biased and correlated due to the natural process that generates them, or because of the fact that an adversary learns some partial information about a random string in cryptographic applications.  

To measure the amount of randomness in a weak random source (a random variable) $X$, we use the standard definition of \emph{min-entropy}: $H_\infty(X)=\min_{x \in \supp(X)}\log_2(1/\Pr[X=x])$. If $X \in \zo^n$, we say $X$ is an $(n,H_\infty(X))$-source, or simply an $H_\infty(X)$-source if $n$ is clear from context. We also say $X$ has \emph{entropy rate} $H_\infty(X)/n$. Ideally, one would like to construct deterministic extractors for all $(n, k)$ sources when $k$ is not too small. However, this is well known to be impossible, even if one only desires to extract one bit and $k$ is as large as $n-1$. Thus, to allow randomness extraction one has to put additional restrictions on the source.

Historically, many different models of randomness extractors have been studied. For example, if one gives the extractor an additional independent short uniform random seed, then there exist extractors that work for any $(n, k)$ source. Such extractors, first introduced by Nisan and Zuckerman \cite{NisanZ96}, are known as \emph{seeded extractors}. These extractors have found wide applications, and by now we have almost optimal constructions (e.g., \cite{LuRVW03, guv, DvirW08, DvirKSS09}) after a long line of research. 

However, seeded extractors may not be applicable in situations where the short uniform random seed is either not available (e.g., in cryptography) or cannot be simulated by cycling over all possible choices. For these applications, one needs \emph{deterministic extractors} or \emph{seedless extractors}, and many different models have also been studied in this setting. These include for example extractors for independent sources \cite{ChorG88, BarakIW04, BarakKSSW05, Raz05, Bourgain05, Rao06, BarakRSW06, Li11b, Li:FOCS:12, Li13a, Li13b, Li:FOCS:15, Cohen15, DBLP:conf/stoc/Cohen16, ChattopadhyayZ:Annals:19, Li:FOCS:16, Coh16b, CL16, Coh16, BADTS:STOC:17, Cohen16, DBLP:conf/stoc/Cohen17, Li17, Li19, lewko_2019, Li:FOCS:23}, bit fixing sources \cite{CGHFRS85, KZ06, GRS06, Rao09}, affine sources \cite{GR08,Bourgain07,Rao09,Yehudayoff11,BK12,Shatiel11b,Li:CCC:11,Li:FOCS:16, ChattopadhyayGL21, Li:FOCS:23}, samplable sources \cite{TV00, Viola14}, interleaved sources \cite{RY08, ChattopadhyayZ:Annals:19}, and small-space sources \cite{KRVZ06}. We define deterministic extractors below.
\BD
Let $\calX$ be a family of distribution over $\zo^n$. A function $\Ext : \zo^n \to \zo^m$ is a deterministic extractor for $\calX$ with error $\eps$ if for every distribution $X \in \calX$, we have 
\[ \Ext(X) \approx_{\eps} U_m,\]
where $U_m$ stands for the uniform distribution over $\zo^m$, and $\approx_{\eps}$ means $\eps$-close in statistical distance. We say $\Ext$ is explicit if it is computable by a polynomial-time algorithm.
\ED
Among these models, two of the most well-studied are extractors for independent sources and affine sources. This is in part due to their connections to several other areas of interest. For example, extractors for independent sources are useful in distributed computing and cryptography with imperfect randomness \cite{KalaiLRZ08, KalaiLR09}, and give explicit constructions of Ramsey graphs; while affine sources generalize bit-fixing sources and extractors for affine sources have applications in exposure-resilient cryptography \cite{CGHFRS85, KZ06} as well as Boolean circuit lower bounds \cite{DK11, FGHK15,LY22}.

Using simple probabilistic arguments, one can show that there exist extractors for two independent $(n, k)$ sources with $k=\log n+O(1)$, which is optimal up to the constant $O(1)$. The first explicit construction of two-source extractors was given by Chor and Goldreich \cite{ChorG88}, which achieves $k> n/2$. Following a long line of research and several recent breakthroughs,  we now have explicit constructions of two-source extractors for entropy $k \approx 4n/9$ with error $\eps=2^{-\Omega(n)}$~\cite{lewko_2019}, for entropy $k=\polylog(n)$ with error $\eps=1/\poly(n)$~\cite{ChattopadhyayZ:Annals:19}, and for entropy $k=O(\log n)$ with constant error \cite{Li:FOCS:23}. Similarly, for affine sources which are uniform distributions over some unknown affine subspace over the vector space $\F^n_2$, \footnote{In this paper we focus on the case where the field is $\F_2$, for larger fields there are affine extractors with better parameters.} one can show the existence of extractors for entropy $k=O(\log n)$, which is also optimal up to the constant $O(1)$. Regarding explicit constructions, we have affine extractors for entropy $k =\delta n$ with error $\eps=2^{-\Omega(n)}$ for any constant $\delta>0$ \cite{Bourgain07, Yehudayoff11, Li:CCC:11}, for entropy $k=\polylog(n)$ with error $\eps=1/\poly(n)$ \cite{Li:FOCS:16}, and for entropy $k=O(\log n)$ with constant error \cite{Li:FOCS:23}. 

In the past decade or so, a new kind of extractors, known as \emph{non-malleable extractors}, has gained a lot of attention. These extractors are motivated by cryptographic applications. Informally, the setting is that an adversary can tamper with the inputs to an extractor in some way, and the non-malleable extractor guarantees that the output of the extractor is close to uniform even conditioned on the output of the extractor on the tampered inputs. The most well-studied non-malleable extractors include seeded non-malleable extractors \cite{DW09}, two-source non-malleable extractors \cite{CG14b}, and affine non-malleable extractors \cite{ChattopadhyayL:STOC:17}. These non-malleable extractors have wide applications in cryptography, such as privacy amplification with an active adversary \cite{DW09} and non-malleable codes \cite{DPW10}. Furthermore, they turn out to have surprising connections to the constructions of standard extractors. Indeed, starting from the work of Li \cite{Li:FOCS:12} which showed a connection between seeded non-malleable extractors and two-source extractors, these non-malleable extractors have played key roles, and now become important cornerstones in the recent series of breakthroughs that eventually lead to explicit constructions of two-source and affine extractors for asymptotically optimal entropy. In a more recent line of work \cite{GryaznovPT:CCC:2022, ChattopadhyayL:ccc:2023, LiZ23}, a special case of affine non-malleable extractors known as \emph{directional affine extractors} is also shown to give strong lower bounds for certain read-once branching programs with linear queries, which generalize both standard read-once branching programs and parity decision trees. Given these applications, non-malleable extractors have become important objects that deserve to be studied on their own. We now define tampering functions and two kinds of non-malleable extractors below.
\begin{definition}[Tampering Function]
For any function $f:S \rightarrow S$, We say $f$ has no fixed points if $f(s) \neq s$ for all $s \in S$. For any $n>0$, let $\mathcal{F}_n$ denote the set of all functions $f: \{ 0,1\}^n \rightarrow \{0,1\}^n$. Any subset of $\mathcal{F}_n$ is a family of tampering functions. 
\end{definition}

\begin{definition}[\cite{CG14b}] 
A function $\tnmext : (\{ 0,1\}^{n})^2 \rightarrow \{ 0,1\}^m$ is a $(k_1, k_2, \eps)$ two-source non-malleable extractor, if it satisfies the following property: Let $X, Y$ be two independent, $(n, k_1)$ and $(n, k_2)$ sources, and $f, g : \zo^n \to \zo^n$ be two arbitrary tampering functions such that at least one of them has no fixed point,\footnote{We say that $x$ is a fixed point of a function $f$ if $f(x)=x$.} then  $$ |\tnmext(X, Y) \circ \tnmext(f(X), g(Y)) - U_m \circ \tnmext(f(X), g(Y))| < \eps.$$ 
\end{definition}


\BD [\cite{ChattopadhyayL:STOC:17}]
A function $\anmExt:\zo^n \rightarrow \zo^m$ is a $(k, \eps)$ affine non-malleable extractor if for any affine source $X$ with entropy at least $k$ and any affine function $f : \zo^n \rightarrow \zo^n$ with no fixed point, we have $$| \anmExt(X) \circ \anmExt(f(X)) - U_m \circ \anmExt(f(X))| \leq \eps.$$
\ED

Using the probabilistic method, one can also prove the existence of these non-malleable extractors with excellent parameters. For example, \cite{CG14b} showed that two-source non-malleable extractors exist for $(n, k)$ sources when $k \geq m+\frac{3}{2}\log (1/\eps)+O(1)$ and $k \geq \log n+O(1)$. Similarly, it can be also shown that affine non-malleable extractors exist for entropy $k \geq 2m+2\log(1/\eps)+\log n+O(1)$.

However, constructing explicit non-malleable extractors appears to be significantly harder than constructing standard extractors, despite considerable effort. Indeed, even for seeded non-malleable extractors, the initial explicit constructions \cite{DodisLWZ:SJoC:14,CRS11,Li12a} only work for sources with entropy rate $>1/2$, and it was not until \cite{ChattopadhyayGL:STOC:16} that explicit seeded non-malleable extractors for sources with poly-logarithmic entropy are constructed. After a long line of research \cite{DodisLWZ:SJoC:14,CRS11,Li12a,Li:FOCS:12, ChattopadhyayGL:STOC:16, Coh15nm, Coh16a, CL16, ChattopadhyayL:STOC:17, Coh16, Cohen16, DBLP:conf/stoc/Cohen17, Li17, Li19, Li:FOCS:23}, an asymptotically optimal seeded non-malleable extractor is finally constructed in \cite{Li:FOCS:23}. On the other hand, the situation for two-source non-malleable extractors and affine non-malleable extractors is much worse, where the best-known constructions in \cite{Li:FOCS:23} only achieve entropy $k > 2n/3$ and $k \geq (1-\gamma)n$ for a small constant $\gamma>0$. This is in sharp contrast to the constructions of standard two-source and affine extractors, where explicit constructions can work for entropy $k=\polylog(n)$ with polynomially small error \cite{ChattopadhyayZ:Annals:19, Li:FOCS:16}, and for entropy $k=O(\log n)$ with constant error \cite{Li:FOCS:23}.

\subsection{Our Results}
In this paper, we study two-source and affine non-malleable extractors for small entropy.\ Our main results give explicit constructions of such non-malleable extractors that essentially match their standard counterparts in the small entropy regime. Specifically, we give explicit two-source and affine non-malleable extractors for $\polylog(n)$ entropy with polynomially small error and for $O(\log n)$ entropy with constant error.\ We have the following theorems.

\begin{theorem} There exists a constant $C>1$ such that for any $k\ge \log^C n$, there exists an explicit construction of a $(k, k, n^{-\Omega(1)})$ two-source non-malleable extractor with output length $\Omega(k)$.
\end{theorem}

\begin{theorem}
There exists a constant $C>1$ such that for any $k\ge \log^C n$, there exists an explicit construction of a $(k, n^{-\Omega(1)})$ affine non-malleable extractor with output length $k^{\Omega(1)}$.
\end{theorem}

\begin{theorem}
There exists a constant $c>1$ such that for any $k\ge c\log n$, there exists an explicit construction of a $(k,k,O(1))$ two-source non-malleable extractor with output length $1$.
\end{theorem}

\begin{theorem}
There exists a constant $c>1$ such that for any $k\ge c\log n$, there exists an explicit construction of a $(k,O(1))$ affine non-malleable extractor with output length $1$.
\end{theorem}

\begin{remark}
    The output length in the two theorems for entropy $k\ge c\log n$ can be extended to a constant number by using the standard XOR lemma and previous techniques (e.g., those in \cite{Li:FOCS:16}). Furthermore, our constructions can also be extended to handle multiple tampering functions as in \cite{ChattopadhyayGL:STOC:16}. For simplicity, we omit the details here.
\end{remark}

The following tables summarize our results compared to some of the best previous constructions.

\begin{table}[H]
\centering
\begin{tabular}{ |c||c|c|c|c|  }
 \hline
 Two-source Non-malleable Extractor & Entropy $k_1$ & Entropy $k_2$ & Output $m$ & Error $\eps$ \\
 \hline
 \cite{ChattopadhyayGL:STOC:16} & $n-n^{\gamma}$ & $n-n^{\gamma}$ & $n^{\Omega(1)}$ & $2^{-n^{\Omega(1)}}$ \\
 \hline
 \cite{Li19} & $(1-\gamma)n$ & $(1-\gamma)n$ & $\Omega(n)$ & $2^{-\Omega(\frac{n\log\log n}{\log n})}$ \\
 \hline
 \cite{AggarwalCO:Crypto:23} & $(\frac{4}{5}+\delta)n$ & $ \log^C n$ & $\Omega(\min\cbra{k_1,k_2})$ & $2^{-\min\cbra{k_1,k_2}^{\Omega(1)}}$ \\
 \hline
 \cite{Li:FOCS:23} & $(\frac{2}{3}+\gamma)n$ & $k = O(\log n)$ & $\Omega(k)$ & $2^{-\Omega(k)}$ \\
 \hline
 Algorithm~\ref{alg:tnmExt-pe}    & $k \ge \polylog(n)$ & $k \ge \polylog(n)$ & $\Omega(k)$ & $n^{-\Omega(1)}$ \\
 \hline
 Algorithm~\ref{alg:tmmExt-ce} & $O(\log n)$ & $O(\log n)$ & 1 & $O(1)$\\
 \hline
\end{tabular}
\end{table}

\begin{table}[H]
\centering
\begin{tabular}{ |c||c|c|c|}
 \hline
 Affine Non-malleable Extractor & Entropy $k$ & Output $m$ & Error $\eps$\\
 \hline
 \cite{ChattopadhyayL:STOC:17}   & $n-n^{\delta}$ for some constant $\delta\in (0,1)$    & $n^{\Omega(1)}$ &   $2^{-n^{\Omega(1)}}$ \\
 \hline
 \cite{Li:FOCS:23} &  $(1-\gamma)n$, $\gamma<1/1000$   & $\Omega(n)$   & $2^{-\Omega(n)}$ \\
 \hline
 Algorithm~\ref{alg:anmExt-pe}    & $\polylog(n)$ & $k^{\Omega(1)}$ & $n^{-\Omega(1)}$ \\
 \hline
 Algorithm~\ref{alg:anmExt-ce}   & $O(\log n)$ & 1 & $O(1)$ \\
 \hline
\end{tabular}
\end{table}

Our results thus significantly improve the entropy requirement of previous non-malleable extractors. As a comparison, we list below the best-known explicit two-source extractors and affine extractors for small entropy.

\begin{table}[H]
\centering
\begin{tabular}{ |c||c|c|c|}
 \hline
 Two-source Extractor & Entropy $k$ & Output $m$ & Error $\eps$\\
 \hline
 \cite{ChattopadhyayZ:Annals:19} & $\polylog(n)$ & $1$ & $n^{-\Omega(1)}$ \\
 \hline
 \cite{Mek:SODA:17,Li:FOCS:16,ChattopadhyayL:FOCS:16} & $\polylog(n)$  &  $k^{\Omega(1)}$  & $n^{-\Omega(1)}$ \\
 \hline
 \cite{BADTS:STOC:17} & $O(\log n 2^{O(\sqrt{\log\log n})})$ & $1$ &  $O(1)$\\
 \hline
 \cite{Cohen16} & $O(\log n(\log\log n)^{O(1)})$ & $1$ &  $O(1)$\\
 \hline
 \cite{Li:STOC:17} & $O(\log n \log\log n)$ & $1$ &  $O(1)$\\
 \hline
 \cite{Li19} & $O(\log n\frac{\log\log n}{\log\log\log n})$ & $1$ &  $O(1)$\\
 \hline
 \cite{Li:FOCS:23} & $O(\log n)$ & $1$ &  $O(1)$\\
 \hline
\end{tabular}
\end{table}

\begin{table}[H]
\centering
\begin{tabular}{ |c||c|c|c|}
 \hline
 Affine Extractor & Entropy $k$ & Output $m$ & Error $\eps$\\
 \hline
 \cite{Li:FOCS:16} & $\polylog(n)$    & $k^{\Omega(1)}$ &   $n^{-\Omega(1)}$ \\
 \hline
 \cite{ChattopadhyayGL21} & $O(\log n\log\log n\log\log\log^6 n)$ & $1$ &  $O(1)$\\
 \hline
 \cite{ChattopadhyayL:STOC:22}    & $O(\log n \log\log n \log\log\log^3 n)$ & $1$ & $O(1)$ \\
 \hline
 \cite{Li:FOCS:23} &  $O(\log n)$ & $1$ & $O(1)$ \\
 \hline
\end{tabular}
\end{table}

It can be seen that the parameters of our two-source and affine non-malleable extractors essentially match those of standard two-source and affine extractors for small entropy. We also point out that the error of our non-malleable extractors is the best one can hope for without first improving the error of standard two-source and affine extractors for small entropy, since the non-malleable extractors are stronger versions of extractors, and in particular, they are themselves two-source and affine extractors. Finally, given that our constructions use many of the key components in the constructions of standard extractors, we believe that any future techniques that improve the error of standard two-source and affine extractors for small entropy (e.g., to negligible error) are also likely applicable to our constructions to get the same improvement on the error of two-source and affine non-malleable extractors.

\subsection{Applications to Lower bounds for Read-Once Linear Branching Programs}
Our affine non-malleable extractors have applications in proving average-case hardness against read-once linear branching programs (ROLBPs). This computational model was recently introduced by Gryaznov, Pudl\'{a}k, and Talebanfard \cite{GryaznovPT:CCC:2022} as a generalization of several important and well-studied computational models such as  decision trees, parity decision trees, and standard read-once branching programs. Roughly, a read-once linear branching program is a branching program that can make linear queries to the input string, while these queries are linearly independent along any path. Formally, we have the following definition.

\begin{definition}[Linear branching program \cite{GryaznovPT:CCC:2022}] A linear branching program on $\F^n_2$ is a directed acyclic graph $P$ with the following properties:
\begin{itemize}
    \item There is only one source $s$ in $P$.
    \item There are two sinks in $P$, labeled with $0$ and $1$ respectively.
    \item Every non-sink node $v$ is labeled with a linear function $\ell_v : \F^n_2 \to \F_2$. Moreover, there are exactly two outgoing edges from $v$, one is labeled with $1$ and the other is labeled with $0$.
\end{itemize}
The size of $P$ is the number of non-sink nodes in $P$. $P$ computes a Boolean function $f : \bin^n \to \bin$ in the following way. For every input $x \in \F^n_2$, $P$ follows the computation path by starting from $s$, and when on a non-sink node $v$, moves to the next node following the edge with label $\ell_v(x) \in \bin$. The computation ends when the path ends at a sink, and $f(x)$ is defined to be the label on this sink.
\end{definition}

\cite{GryaznovPT:CCC:2022} defines two kinds of read-once linear branching programs ($\mathsf{ROLBP}$ for short). 

\begin{definition} \cite{GryaznovPT:CCC:2022}
Given any linear branching program $P$ and any node $v$ in $P$, let $\pre_v$ denote the span of all linear queries that appear on any path from the source to $v$, excluding the query $\ell_v$. Let $\post_v$ denote the span of all linear queries in the subprogram starting at $v$.
\begin{itemize}
    \item A linear branching program $P$ is weakly read-once if for every inner node $v$ of $P$, it holds that $\ell_v \notin \pre_v$.
    \item A linear branching program $P$ is strongly read-once if for every inner node $v$ of $P$, it holds that $\pre_v \cap \post_v=\{0\}$.
\end{itemize}
\end{definition}

Both kinds of $\mathsf{ROLBP}$s generalize the aforementioned computational models, but weakly read-once linear branching programs ($\WROLBP$s) are more flexible than strongly read-once linear branching programs ($\SROLBP$s). As a result, proving lower bounds for $\WROLBP$s turns out to be much harder than for $\SROLBP$s. Indeed, so far we only have non-trivial lower bounds for $\SROLBP$s. To state our results, we use the following definition.

\begin{definition}\cite{ChattopadhyayL:ccc:2023}
    For a Boolean function $f: \bin^n \to \bin$, let $\SROLBP(f)$ denote the smallest possible size of a strongly read-once linear branching program that computes $f$, and $\SROLBP_{\eps}(f)$ denote the smallest possible size of a strongly read-once linear branching program $P$ such that 
    $$\Pr_{x \leftarrow_U \F^n_2}[P(x)=f(X)] \geq \frac{1}{2}+\eps.$$ The definition can be adapted to $\ROBP$s naturally.
\end{definition}

\cite{GryaznovPT:CCC:2022} shows that a stronger version of affine extractors known as \emph{directional affine extractors} give strong average case lower bounds for $\SROLBP$s. They give an explicit construction of directional affine extractors for entropy $k \geq \frac{2n}{3}+c$ with error $\eps \leq 2^{-c}$ for any constant $c>1$, which also implies exponential average-case hardness for $\SROLBP$s of size up to $2^{\frac{n}{3}-o(n)}$. In a follow-up work, Chattopadhyay and Liao \cite{ChattopadhyayL:ccc:2023} used another kind of extractors known as \emph{sumset extractors} \cite{ChattopadhyayL16} to give an alternative average-case hardness for $\SROLBP$s. In particular, they gave an explicit function $\Ext$ such that $\SROLBP_{n^{-\Omega(1)}} (\Ext) \geq 2^{n-\log^{O(1)} n}$. More recently, Li \cite{Li:FOCS:23} gave an improved sumset extractor which in turn yields an explicit function $\Ext$ such that $\SROLBP_{2^{-\Omega(1)}} (\Ext) \geq 2^{n-O(\log n)}$. In these two constructions, the branching program size lower bounds become quite close to optimal (the result of \cite{Li:FOCS:23} is optimal up to the constant in $O(\cdot)$), while the correlation becomes polynomially large or a constant. Another recent work by Li and Zhong \cite{LiZ23} gave explicit directional affine extractor for entropy $k \geq cn(\log\log\log n)^2/\log\log n$ with error $\eps=2^{-n^{\Omega(1)}}$ for some constant $c>1$, which implies exponential average-case hardness for $\SROLBP$s of size up to $2^{n-o(n)}$.

For simplicity, we do not define directional affine extractors here, but just mention that directional affine extractors are a special case of affine non-malleable extractors. Hence, our new constructions of affine non-malleable extractors directly imply improved directional affine extractors, which in turn also give average-case hardness for $\SROLBP$s. Specifically, we have the following theorem.

\begin{theorem}
There exist explicit functions $\anmExt_1$, $\anmExt_2$ such that $\SROLBP_{n^{-\Omega(1)}} (\anmExt_1) \geq 2^{n-\log^{O(1)} n}$ and $\SROLBP_{2^{-\Omega(1)}} (\anmExt_2) \geq 2^{n-O(\log n)}$.
\end{theorem}

These bounds match the previously best-known average-case hardness results for $\SROLBP$s given in \cite{ChattopadhyayL:ccc:2023} and \cite{Li:FOCS:23}, where the branching program size lower bounds are close to optimal, but the explicit functions we use here are different. Specifically, here we use affine non-malleable extractors while \cite{ChattopadhyayL:ccc:2023} and \cite{Li:FOCS:23} use sumset extractors.

\subsection{Technical Overview}
Here we outline the main techniques used in this paper, opting for an informal approach at times for clarity while omitting certain technical details.

We use the standard notation in the literature where a letter with $'$ represents a tampered version. Let $f$ and $g$ denote the tampering functions on $X$ and $Y$ in two-source non-malleable extractors, respectively, and $\A$ be the affine tampering function in affine non-malleable extractors.

Since two-source and affine non-malleable extractors are themselves two-source and affine extractors, our high-level idea is to adapt the constructions of standard extractors for polylogarithmic or logarithmic entropy into the stronger, non-malleable version. Clearly, a direct naive application of standard extractors may not work, since the output on the tampered inputs may be correlated to the output on the original inputs. Below we start with two-source extractors to illustrate our main ideas. Let us first briefly review the constructions of two-source extractors for small entropy. Generally, these extractors are double-layered: the outer layer is a suitable resilient function, which is designed to be an extractor for non-oblivious bit-fixing (NOBF) sources with $t$-wise independent property for some parameter $t$. That is, most of the bits are $t$-wise independently uniform, while the rest of the bits can depend arbitrarily on these bits. Here, the extractor uses a crucial property that bounded independence suffices to work for several resilient functions (or equivalently these functions are \emph{fooled} by bounded independence), such as the derandomized Ajtai-Linial function in~\cite{ChattopadhyayZ:Annals:19} or the Majority function. The inner layer is a transformation that transforms two independent sources into a single NOBF source with the $t$-wise independent property. This step itself utilizes techniques from seeded non-malleable extractors or correlation breakers, which are functions designed to break correlations between random variables. 

To adapt the construction to two-source non-malleable extractors, our first observation is that there is an easy case. Intuitively, this is the case where one input source has large entropy conditioned on the tampered version. For instance, say the source $X$ has high entropy conditioned on every fixing of $X'=f(X)=x'$. Then in the analysis, we can first fix $X'$, and then further fix the tampered output of the extractor, which is now a deterministic function of $Y$ and can be chosen to have a relatively small size. Conditioned on these fixings, we have $X$ and $Y$ are still independent and have good entropy, hence any two-source extractor will give an output that is close to uniform conditioned on the tampered output. 

However, it is certainly possible that the above does not hold. For example, the tampering function $f$ can be an injective function, so that conditioned on any fixing of $X'=f(X)=x'$, we have that $X$ is also fixed. In this case, our observation is that $X'$ itself must also have large entropy (since $f$ injective), therefore we can possibly create structures in the distribution of the tampered version as well. Specifically, our strategy is to modify the inner layer of the two-source extractor while essentially using the same outer layer. For simplicity, let us consider extractors with just one bit of output. A standard approach to show the output bit is close to uniform conditioned on the tampered output, is to show that the parity of these two bits is close to uniform. Since the outer extractor is a resilient function, this suggests to look at the parity of two copies of resilient functions on two correlated distributions.

Now another crucial observation behind our construction is that just like in the construction of standard two-source extractors, for certain resilient functions, the parity of two copies of such functions is still fooled by bounded independence. Thus, if in the inner layer, we can create structures such that the \emph{joint distribution} of the NOBF source and the tampered version has the $t$-wise independent property, then we will be able to show that the extractor is non-malleable. Note that we are now in the case where the tampered sources also have high entropy, which works in our favor since achieving $t$-wise independence requires a certain amount of entropy. However, we cannot simply use previous techniques since the tampered sources are correlated with the original sources. Therefore, we appropriately modify previous constructions of correlation breakers to ensure the $t$-wise independent property in the joint distribution.

Finally, in the actual analysis, we are not guaranteed to be in either case; and it may happen that for some $x'$, conditioned on $X'=x'$ we have that $X$ has large entropy, while for others conditioned on $X'=x'$ we have that $X$ has small entropy. The analysis thus needs a careful interpolation between different cases in terms of a convex combination of subsources. We now elaborate with more details on each of these aspects below.

First we give some notation that will help with our presentation.

\BD[$t$-non-malleable $(k,\eps)$ seeded extractor]\label{def:nmExt}
    A function $\nmExt:\bin^n \times \bin^d \to \bin^m$ is a $t$-non-malleable $(k,\eps)$ extractor if it satisfies the following property: if $X$ is a $(n,k)$-source and $Y$ is uniform on $\bin^d$, and $f_1,\cdots,f_t$ are arbitrary functions from $d$ bits to $d$ bits with no fixed point, then
    \begin{align*}          
        (\nmExt(X,Y),\nmExt(X,f_1(Y)),\cdots,\nmExt(X,f_t(Y)),Y)\approx_\eps (U_m,\nmExt(X,f_1(Y)),\cdots,Y).
    \end{align*}
\ED

We say a distribution or a source $X$ on $n$ bits is $(q, t, \gamma)$ independent if there exists a subset $S \subseteq [n]$ with $|S| \leq q$ such that if we consider the bits of $X$ in $[n] \setminus S$, then every $t$ bits are $\gamma$-close to uniform.

 A $t$-non-malleable $(k,\eps)$ seeded extractor $\nmExt$ with seed length $d$ can be used to generate a $(q, t, \gamma)$ source from a source $X$ with entropy at least $k$ in the following way: cycle over all possible seeds $i$, and for each one output a bit $\nmExt(X,i)$. The output is now $(\sqrt{\eps}D,t+1,t\sqrt{\eps})$ independent with $D=2^d$.

\subsubsection{Taking the parity of two resilient functions}
A Boolean function on $n$ variables is a resilient function if it is nearly balanced, and no small coalition can have a significant influence on the output of the function. Such functions are equivalent to extractors for NOBF sources. The resilient functions that have played a key role in the recent advancement of extractors are the derandomized Ajtai-Linial function in~\cite{ChattopadhyayZ:Annals:19} and the Majority function. The former is a monotone $\ac^0$ function, that is fooled by $\polylog(n)$-wise independence, while the latter is a threshold function that can be fooled by constant-wise independence.

It is not hard to show that the parity of two independent copies of resilient functions is still a resilient function.\ What is left to show is that such a parity can also be fooled by bounded independence. When the resilient function is in $\ac^0$, we observe that the parity of two such functions is also in $\ac^0$, because the parity of two bits can be written as a constant size $\ac^0$ circuit. Therefore, the parity of two derandomized Ajtai-Linial functions is still in $\ac^0$, and can be fooled by $\polylog(n)$-wise independence by Braverman's celebrated result~\cite{Braverman10} on bounded independence fooling $\ac^0$ circuits, together with some standard techniques.

To show that constant-wise independence fools the parity of two Majority functions, we use the work of Gopalan, O'Donnell, Wu, and Zuckerman~\cite{POWZuckerman:CCC:10}, which shows that constant-wise independence fools any function of halfspaces under product distributions, as long as the function can be implemented as a constant size circuit. In our case, this clearly holds since we are just taking the parity of two Majority functions. Using the XOR lemma and previous techniques (e.g., those in \cite{Li:FOCS:16}, our construction can also be extended to output a constant number of bits.

\subsubsection{Generating NOBF sources from the inputs and its tampered counterparts}\label{subsec:nobf} We want to construct a function such that when the tampered sources have sufficient entropy, the joint distribution of the generated bits from the input sources and the tampered sources is $(q,t,\gamma)$ independent for some suitable parameters $q, t$, and $\gamma$. 

The standard approach for two-source extractors, as introduced in~\cite{ChattopadhyayZ:Annals:19}, is to first apply a seeded non-malleable extractor to one source, say $Y$, and then use another source $X$ to sample a small number of bits from the output. However, in our case, this black-box approach does not work since the tampered sources are correlated with the original inputs. Therefore, we have to create some kind of difference between the tampered sources and the original sources, which will enable us to get the desired $(q,t,\gamma)$ independent property. 

To achieve this, we dig into the constructions of seeded non-malleable extractors and existing two-source non-malleable extractors, which roughly go as follows. First, one uses an \emph{advice generator} to create a short string that is different from the tampered version with high probability. Then, conditioned on the fixing of the advice strings, one can argue that the two sources are still independent and have sufficient entropy. At this point one uses a \emph{correlation breaker with advice}, together with the advice strings to compute the output, which is guaranteed to have the non-malleable property. However, the steps of generating advice and subsequent application of correlation breakers require the sources to have very large entropy (e.g., at least $2/3$), which is the main reason that previous two-source non-malleable extractors can only work for large entropy.

To get around this barrier, our approach is to first apply a standard seeded extractor to one source $Y$, and output say $\Omega(k)$ bits where $k$ is the entropy. By cycling over all possible seeds, we potentially get a matrix with $D=2^d$ rows where $d$ is the seed length of the seeded extractor. We then use the other source $X$ to sample a small number ($\poly(n)$) of rows from the output. Now, again a standard argument implies that most of the rows are close to uniform. Since we are in the case where the tampered sources $X', Y'$ have sufficient entropy, this is also true for the tampered version. Note that we haven't achieved the $(q,t,\gamma)$ independent property yet. Our next step is to generate advice from the original sources and the tampered sources. However, in the low-entropy regime, it is hard to generate a single advice for the input sources and the tampered version --- the advice generator requires generating uniform seeds to sample from an encoding of the inputs and it is hard to do so from a slice of the sources which could have zero entropy. Therefore, we generate advice from each row. We can then append the index of this row to the advice. This ensures that the advice strings are both different from the tampered version, and also different between different rows. Now, we can apply existing constructions of correlation breakers with advice, which will ensure that for any $t$ rows in the combined matrix from the original sources and the tampered sources, as long as all these rows have high entropy initially, the joint distribution of the final outputs from the correlation breaker is $\gamma$-close to uniform. 

However, there are additional tricky issues with this approach. First, the correlation breaker requires two independent sources to work, while in our case the outputs in the matrices are already functions of both $X$ and $Y$. Second, the analysis of the correlation breaker usually requires fixing the advice strings first and arguing that the sources still have sufficient entropy, but now since the matrices have $\poly(n)$ rows and the entropy of the sources is just $k=\polylog(n)$, or even $k=O(\log n)$,
 if we fix all the advice strings then conditioned on the fixing the sources may not have any entropy left. Finally, the set of ``good" rows (the rows that are close to uniform after the sampling using $X$) in the matrices depends on the source $X$ and $Y$, and after we fix the advice strings in the analysis, $X$ and $Y$ may have become different, and this could potentially change the set of ``good" rows in the first place.

To solve these issues, we use an argument similar to that in~\cite{Li:FOCS:15}. The idea is that since eventually we only need $(q,t,\gamma)$ independence, in the analysis we can just focus on every subset of $t$ rows from the good rows. In particular, we can set $t$ and the entropy $k$ appropriately, i.e., $t$ is relatively small compared to $k$. This is because we only need $t=\polylog (n)$ to apply the derandomized Ajtai-Linial in~\cite{ChattopadhyayZ:Annals:19} and  $t=O(1)$ to apply the Majority function. 
Now in the analysis, notice that the process of sampling using the source $X$ basically corresponds to $\Ext(Y, \Ext'(X, i))$ where $\Ext, \Ext'$ are two seeded extractors. Thus when $t$ is small, for any subset $T$ with $|T|=t$, we can first fix all $\Ext'(X, i)$ with $i \in T$. By restricting the size of $\Ext'(X, i)$, $X$ still has sufficient entropy conditioned on these fixings, and now the $t$ rows of the outputs $\Ext(Y, \Ext'(X, i))$ are deterministic functions of $Y$, while $X$ and $Y$ are still independent. By restricting the size of the advice strings, we can preserve the above properties when the analysis fixes the advice strings and goes into the correlation breaker. Finally, as in the analysis in~\cite{Li:FOCS:15}, the final error pays a price of a $\poly(n)^t$ factor from a union bound on all possible subsets of size $t$, which is still fine as long as we set $k \gg t \log n$ and use seeded extractors with error $2^{-\Omega(k)}$ in the correlation breaker.

\subsubsection{Convex combination of subsources}\label{subsec:convex-comb} In the above two subsections, we dealt with the case where the tampered sources have sufficient entropy. We now sketch the analysis for the general case. 

Let $\tnmext$ be the two-source non-malleable extractor which works if both $X$ and $X'$ have entropy at least $k_x$, and both $Y$ and $Y'$ have entropy at least $k_y$, with error $\eps/2$ and output length $m$. Now assume $X$ has min-entropy $2k_x+\log(2/\eps)$, and $Y$ has min-entropy $2k_y+\log(2/\eps)$. Further assume without loss of generality that both $X$ and $Y$ are flat sources, i.e., uniform distributions over some unknown subset. The analysis goes by considering the ``heavy" elements in the tampered sources $X'$ and $Y'$. Specifically, for any $x' \in \bits^n$ and $y' \in \bits^n$, we consider the pre-image size of $X'=x'$ and $Y'=y'$. If one of them is large, say without loss of generality that the pre-image size of $X'=x'$ is at least $2^{k_x}$, then $\minH(X|X'=x)\ge k_x$. We can first fix $X'=x'$ and then $\tnmext(x',Y')$, which is a deterministic function of $Y$ now conditioned on the fixing of $X'=x'$. Since $\tnmext(x',Y')$ is short compared to $\minH(Y)$, conditioned on these fixings we have that $X$ and $Y$ are still independent and have sufficient entropy, so $\tnmext(X,Y)$ is close to uniform because $\tnmext$ is itself a two-source extractor. Note that we have already fixed $\tnmext(x',Y')$, and thus $\tnmext$ is indeed non-malleable in this case. 

Next, consider the set of all the $x'$ whose pre-image size under $f$ is at most $2^{k_x}$, and call it $\BAD_X$. If the total probability mass of these $x'$ is at most $\eps/2$, then we can just ignore them (and the corresponding $x$ in the support of $X$) since this only adds an extra error of $\eps/2$. Similarly, we can also ignore the set of all the $y'$ whose pre-image size under $g$ is at most $2^{k_y}$ (call it $\BAD_Y$), if the total probability mass of these $y'$ is at most $\eps/2$. In either case, we are done. Otherwise, the subsource of $X'$ formed by all the $x' \in \BAD_X$ has min-entropy at least $-\log (2^{k_x}/(\eps 2^{2k_x}/\eps))=k_x$, and the corresponding subsource of $X$ has min-entropy at least $2k_x$. Similarly, the subsource of $Y'$ formed by all the $y' \in \BAD_Y$ has min-entropy at least $k_y$, and the corresponding subsource of $Y$ has min-entropy at least $2k_y$. In this case, both sources and their tampered versions have sufficient entropy, thus by the analysis before, $\tnmext$ is also a non-malleable extractor.

Since $X$ is just a convex combination of subsources $(X \mid X'=x' \in \BAD_x)$ and $\{(X \mid X'=x' \notin \BAD_x)\}$, and the same is true for $Y$, the correctness of $\tnmext$ follows.

Finally, we note that we can modify the two-source non-malleable extractor to output $k^{\Omega(1)}$ bits, by using a similar approach based on the XOR lemma as in \cite{Li:FOCS:16}. Then, since the two-source non-malleable extractor is strong, we can further apply a standard seeded extractor to increase the output length to $\Omega(k)$.

\subsubsection{Affine non-malleable extractors}
Our construction of affine non-malleable extractors roughly follows the same ideas. The difference is that now we do not have access to two independent sources, but the affine source itself has nice structures and the tampering function is affine. Thus, by using appropriate linear seeded extractors as in previous works, certain parts of the affine source behave like independent sources. Therefore, we can suitably adapt our construction of two-source non-malleable extractors to affine non-malleable extractors. One particularly nice property of affine sources is that when applying a strong linear seeded extractor on an affine source, the output on most seeds \emph{is uniform}. This implies that we can generate from an affine source a somewhere random source with no error. In the two-source case, we cannot analyze the generation of NOBF source directly from the definition of the correlation breaker (and have to resort to additional techniques as mentioned in previous paragraph~\ref{subsec:nobf}) due to the error of the somewhere random source. In the affine case, there is no such concern. Therefore, we can argue that we obtain a NOBF source directly from the definition of affine correlation breaker, as in prior works on affine extractors for small entropy, e.g.,~\cite{ChattopadhyayGL21}.

\subsection{Organization of the Paper}
The rest of the paper is organized as follows. In Section~\ref{sec:prelim} we give some preliminary knowledge. In section~\ref{sec:xor}, we present the $\oplus$ generalization of Resilient functions. In Section~\ref{sec:advGen}, we present advice generators for small entropy. In Section~\ref{sec:nmext-polylog}, we present the two-source non-malleable extractor and affine non-malleable extractor for polylog entropy with polynomially small error. In Section~\ref{sec:nmext-const}, we present the two-source non-malleable extractor and affine non-malleable extractor for logarithmic entropy with constant small error. We conclude the paper and present some open problems in Section~\ref{sec:open}.  

\section{Preliminaries} \label{sec:prelim}
We often use capital letters for random variables and corresponding small letters for their instantiations. Let $s,t$ be two integers, $\cbra{V_{1,1},V_{1,2},\cdots,V_{1,t},V_{2,1},V_{2,2},\cdots, V_{2,t},\cdots,V_{s,1},V_{s,2},\cdots,V_{s,t}}$ be a set of random variables. We use $V_{i,[t]}$ to denote the subset $\cbra{V_{i,1},\cdots,V_{i,t}}$.

Let $|S|$ denote the cardinality of the set~$S$. 

For $\ell$ a positive integer,
$U_\ell$ denotes the uniform distribution on $\zo^\ell$. When used as a component in a vector, each $U_\ell$ is assumed independent of the other components.

All logarithms are to the base 2.

Let $v\in \bin^n$, we use $\abs{v}$ to denote the hamming weight of $v$.

For any strings $x_1,x_2\in\bin^*$, we use $x_1 \circ x_2$ to denote the concatenation of $x_1$ and $x_2$.

For any $n\in \N$, any strings $x_1\in \bin^n$, $x_2\in \bin^{\log(n)}$, we use $x_{1\mid x_2}$ to denote the bit of $x_1$ indexed by $x_2$.

\subsection{Probability Distributions, Entropy, and (Sub)sources}
\begin{definition} [statistical distance]Let $W$ and $Z$ be two distributions on
a set $S$. Their \emph{statistical distance} (variation distance) is
\[\Delta(W,Z) \eqdef \max_{T \subseteq S}(|W(T) - Z(T)|) = \frac{1}{2}
\sum_{s \in S}|W(s)-Z(s)|~.
\]
\end{definition}

We say $W$ is $\eps$-close to $Z$, denoted $W \approx_\eps Z$, if $\Delta(W,Z) \leq \eps$. Let $V$ also be a distribution on the set $S$. We sometimes use $W \approx_\eps Z \mid V$ as a shorthand for $(W,V)\approx_\eps(Z,V)$. We will use these two notations interchangeably throughout the paper.
For a distribution $D$ on a set $S$ and a function $h:S \to T$, let $h(D)$ denote the distribution on $T$ induced by choosing $x$ according to $D$ and outputting $h(x)$.

\BL \label{lem:sdis}
For any function $\alpha$ and two random variables $A, B$, we have $\Delta(\alpha(A), \alpha(B)) \leq \Delta(A, B)$.
\EL

\begin{definition}[min-entropy]
    The \emph{min-entropy} of a random variable $X$ is defined as 
    \[
        H_\infty(X) = \min_{x\in\Supp(X)}\cbra{-\log_2\Pr[X=x]}~.
    \]
\end{definition}
For a random variable $X\in \bin^n$, we say it is an $(n,k)$-source if $H_\infty(X) \ge k$. When $n$ is understood from the context we simply say that $X$ is a $k$-source.
The entropy rate of $X$ is defined as $H_\infty(X)/n$.

\begin{definition}[subsource]
    Let $X$ be an $n$-bit source in some probability space. We say that an event $A$ is determined by $X$ if there exists a function $f:\bin^n \to \bin$ such that $A=\{f(X)=1\}$. We say $X_0$ is a subsource of $X$ if there exists an event $A$ that is determined by $X$ such that $X_0=(X\mid A)$.
\end{definition}

\subsection{Somewhere Random Sources, Extractors and Samplers}

\begin{definition} [Somewhere random sources] \label{def:SR} A source $X=(X_1, \cdots, X_t)$ is $(t \times r)$
  \emph{somewhere-random} (SR-source for short) if each $X_i$ takes values in $\bits^r$ and there is an $i$ such that $X_i$ is uniformly distributed.
\end{definition}

\BD
An elementary somewhere-k-source is a vector of sources $(X_1, \cdots, X_t)$, such that some $X_i$ is a $k$-source. A somewhere $k$-source is a convex combination of elementary somewhere-k-sources.
\ED

\BD
A function $C: \bits^n \times \bits^d \to \bits^m$ is a $(k \to \ell, \eps)$-condenser if for every $k$-source $X$, $C(X, U_d)$ is $\eps$-close to some $\ell$-source. When convenient, we call $C$ a rate-$(k/n \to l/m, \eps)$-condenser.   
\ED

\BD
A function $C: \bits^n \times \bits^d \to \bits^m$ is a $(k \to \ell, \eps)$-somewhere-condenser if for every $k$-source $X$, the vector $(C(X, y)_{y \in \bits^d})$ is $\eps$-close to a somewhere-$\ell$-source. When convenient, we call $C$ a rate-$(k/n \to \ell/m, \eps)$-somewhere-condenser.   
\ED

\begin{definition}[Strong seeded extractor]\label{def:strongext}
A function $\Ext : \bits^n \times \bits^d \rightarrow \bits^m$ is  a \emph{strong $(k,\eps)$-extractor} if for every source $X$ with min-entropy $k$
and independent $Y$ which is uniform on $\zo^d$,
\[ (\Ext(X, Y), Y) \approx_\eps (U_m, Y).\]
\end{definition}


\BL[Strong seeded extractor with deficient seed~\cite{ChattopadhyayGL:STOC:16}]\label{lemma:sext-deficient}
Let $\Ext:\bin^n \times \bin^d \to \bin^m$ be a strong seeded extractor for min-entropy $k$, and error $\eps$. Let $X$ be a $(n,k)$-source and let $Y$ be a source on $\bin^d$ with min-entropy $d-\lambda$. Then,
\[
    \abs{\Ext(X,Y)\circ Y-U_m\circ Y} \le 2^\lambda \eps.
\]
\EL

\BT[\cite{GuruswamiUV:JACM:09}]\label{thm:GUV09}
    For any constant $\alpha>0$, and all integers $n,k>0$ there exists a polynomial time computable $(k,\eps)$-strong seeded extractor $\Ext:\bin^n \times \bin^d \to \bin^m$ with $d = O(\log(n/\eps))$ and $m=(1-\alpha)k$. 
\ET

The following is an explicit affine extractor constructed by Bourgain.
\BT[\cite{Bourgain07}]\label{thm:BAExt}
For all $n,k>0$ and any constant $\delta>0$ such that $k\ge \delta n$, there exists an explicit affine extractor $\AExt:\bin^n \to \bin^m$, $m = \Omega(k)$, for min-entropy $k$ with error $2^{-\Omega(k)}$.

\ET

\BD\label{def:samp}
    $\Samp:\bin^n\times[D]\to \bin^m$ is an $(\eps,\delta)$-sampler for min-entropy $k$ if for every test $T\subseteq \bin^m$ and every $(n,k+\log(1/\delta))$-source $X$,
    \[
        \Pr_{x\sim X}\sbra{\Pr_{y\sim [D]}[\Samp(x,y)\in T]\ge \mu_T - \eps}\le \delta,
    \]
    where $\mu_T=|T|/2^m$.
\ED

\BT[\cite{Zuc:RANDOM:97}]\label{thm:samp}
    Let $\Ext:\bin^n \times \bin^d \to \bin^m$ be a seeded extractor for min-entropy $k$ and error $\eps$. Let $\bin^d = \cbra{r_1,\cdots,r_D}$, $D=2^d$. Define $\Samp(x) = \cbra{\Ext(x,r_1),\cdots,\Ext(x,r_D)}$. Let $X$ be an $(n,2k)$-source. Then for any set $R\subseteq \bin^m$,
    \begin{align*}
        \Pr_{x\sim X}[||\Samp(x)\cap R|-\mu_RD|>\eps D]<2^{-k}
    \end{align*}
    where $\mu_R = |R|/2^m$.
\ET

The following seeded extractor by Zuckerman~\cite{Zuckerman:TOC:07} achieves seed length $\log(n)+O(\log(1/\eps))$ to extract from any source with constant min-entropy rate, matching the probabilistic bound. 

\BT[\cite{Zuc:RANDOM:97}]\label{thm:sext-optimal}
For all $n>0$ and constants $\alpha,\delta,\eps>0$ there exists an efficient construction of a $(k=\delta n,\eps)$-strong seeded extractor $\Ext:\bin^n\times \bin^d\to \bin^m$ with $m\ge (1-\alpha)k$ and $D=2^d=O(n)$.
\ET

The extractor is used to obtain the following sampler.

\BD[Averaging sampler~\cite{Vadhan:JoC:04}]\label{def:avg-samp}
    A function $\Samp:\bin^r \to [n]^t$ is a $(\mu,\theta,\gamma)$ averaging sampler if for every function $f:[n]\to [0,1]$ with average value $\frac{1}{n}\sum_i f(i)\ge \mu$, it holds that
    \[
        \Pr_{i_1,\cdots,i_t\leftarrow \Samp(U_r)}\sbra{\frac{1}{t}\sum_i f(i) < \mu-\theta} \le \gamma.
    \]
\ED

\BT[\cite{Vadhan:JoC:04}]\label{thm:avg-samp}
    For every $0<\theta<\mu<1,\gamma>0$, and $n\in \N$, there is an explicit $(\mu,\theta,\gamma)$ averaging sampler $\Samp:\bin^r \to [n]^t$ that uses
    \begin{itemize}
        \item $t$ distinct samples for any $t\in [t_0,n]$, where $t_0 = O(\frac{1}{\theta^2}\log(1/\gamma))$, and 
        \item $r=\log(n/t)+\log(1/\gamma)\poly(1/\theta)$ random bits.
    \end{itemize}
\ET

\begin{remark}\label{rmk:samp-deficient}
    Applying a source on $\bin^r$ with $r-\lambda$ entropy to $\Samp$ from Theorem~\ref{thm:avg-samp} will result in a multiplicative factor of $2^\lambda$ to $\gamma$ by Lemma~\ref{lemma:sext-deficient} and Theorem~\ref{thm:samp}.
\end{remark}

A generalization of sampler is the so-called somewhere random sampler.

\BD[\cite{BADTS:STOC:17}]\label{def:srs}
$\srSamp:\bin^n \times [D] \times [C] \to \bin^m$ is a $(\eps,\delta)$-somewhere random sampler for entropy $k$ if for every set $T\subseteq \bin^m$ s.t. $|T|\le \eps 2^m$ and every $(n,k)$-source $X$,
\begin{align*}
    \Pr_{x\sim X}\sbra{\Pr_{y\sim [D]}\sbra{\forall z \in [C]\;\srSamp(x,y,z)\in T}>2\eps} \le \delta.
\end{align*}
We say $\srSamp$ is linear if $\srSamp(\cdot,y,z)$ is linear for every $y\in [D], z\in [C]$.
\ED

A somewhere random sampler is usually constructed by composing a sampler with a disperser.

\begin{definition}\label{def:Gamma}
    A function $\Gamma:[N] \times [D] \to [M]$ is a $(K,\eps)$-disperser if for every set $X\subseteq [N]$ with $|X|\ge K$, the set $\Gamma(X):=\{\Gamma(x,y) \mid x\in X,y\in [D]\}$ satisfies
    \begin{align*}
        |\Gamma(X)|\ge \eps M.
    \end{align*}
\end{definition}

\begin{lemma}[\cite{Zuckerman:TOC:07}]\label{lemma:disperser-Zuc07}
    For every constant $\delta>0$ and $\eps = \eps(n)>0$, there exists an efficient family of $(K = N^\delta, \eps)$-disperser $\Gamma:[N = 2^n] \times [D] \to [M]$ such that $D = O(\frac{n}{\log(1/\eps)})$ and $M = \sqrt{K}$.
\end{lemma}

\subsection{Average Conditional Min-Entropy and Average-Case Seeded Extractors}

\begin{definition}[Average conditional min-entropy]
The \emph{average conditional min-entropy} is defined as 
\begin{align*}
    \widetilde{H}_\infty(X\mid W) &= -\log\pbra{\E_{w\leftarrow W}\sbra{\max_x \Pr[X=x\mid W=w]}} \\
    & = -\log\pbra{\E_{w\leftarrow W}\sbra{2^{-H_\infty(X \mid W=w)}}}.
\end{align*}
\end{definition}

\begin{lemma}[\cite{DodisORS08}]\label{lemma:minH-avgH}
For any $s>0$, $\Pr_{w\leftarrow W}[H_\infty(X\mid W=w)\ge \widetilde{H}_\infty(X\mid W)-s]\ge 1-2^{-s}$.
\end{lemma}

\begin{lemma}[\cite{DodisORS08}]\label{lemma:avgH-minH}
If a random variable $B$ has at most $2^\ell$ possible values, then $\widetilde{H}_\infty(A\mid B)\ge H_\infty(A)-\ell$.
\end{lemma}

\begin{lemma}[\cite{DodisORS08}]\label{lemma:avg-ext}
For any $\delta>0$, if $\Ext$ is a $(k,\eps)$ extractor, then it is also a $(k+\log(1/\delta),\eps+\delta)$ average case extractor.
\end{lemma}

\subsection{Linear Seeded Extractors}

We say that the function is a linear strong seeded
extractor if the seeded extractor $\Ext(\cdot,u)$ is a linear function over $\F_2$, for every $u\in\bin^d$.

The following linear seeded extractor is used in advice generator for affine source as it achieves simultaneously $O(\log(1/\eps))$ entropy, $O(\log(1/\eps))$ seed length for $\eps=n^{-\omega(1)}$ or $\eps=n^{-\Omega(1)}$.
\begin{lemma}[\cite{ChattopadhyayGL21}]\label{lemma:log_n_lext}
There exists a constant $C_{\ref{lemma:log_n_lext}}$ such that for every $m\in \N$ and $\eps>0$, there exists an explicit $(C_{\ref{lemma:log_n_lext}}(m+\log(1/\eps)),\eps)$-linear strong seeded extractor $\LExt:\bin^n\times\bin^d\to \bin^m$ s.t. $d=O(m+\log(n/\eps))$.
\end{lemma}

In constructing non-malleable extractors for $\polylog(n)$ entropy, we will generate a somewhere random source with a strong linear seeded extractor. The somewhere random source will have $D=2^d$ rows where $d$ is the seed length of the extractor. The number of rows can be at most $\poly(n)$. For that, we need $d$ to have optimal seed length $O(\log(n))$ and that is achieved by the following construction.

\begin{theorem}[\cite{Li:FOCS:16}]\label{thm:LExt-optimal-seedlength}
    There exists a constant $c>1$ such that for every $n,k \in \N$ with $c\log^8 n \le k \le n$, and $\eps \ge n^{-2}$, there is an explicit $(k,\eps)$ strong linear seeded extractor $\LExt:\bin^n \times \bin^d \to \bin^m$ with $d = O(\log n)$ and $m = \sqrt{k}$.
\end{theorem}

The following strong linear seeded extractor is used to increase the output length of two-source non-malleable extractors for $\polylog(n)$ entropy as it achieves $\Omega(k)$ output length with $O(\log^2(n))$ seed length.

\begin{theorem}[\cite{Trevisan01,RazRV02}]\label{thm:LExt-Trev}
    For every $n,k,m\in \N$ and $\eps>0$, with $m\le k\le n$, there exists an explicit strong linear seeded extractor $\LExt:\bin^n\times \bin^d \to \bin^m$ for min-entropy $k$ and error $\eps$, where $d = O\pbra{\frac{\log^2(n/\eps)}{\log(k/m)}}$.
\end{theorem}


\subsection{The Structure of Affine Sources}\label{subsec:affine}

In this paper, affine sources encompass uniform distributions over linear subspaces and by affine functions we sometimes mean affine-linear functions.

\begin{definition}[Affine source]
    Let $\F_q$ be the finite field with $q$ elements. Denote by $\F_q^n$ the $n$-dimensional vector space over $\F_q$. A distribution $X$ over $\F_q^n$ is an $(n,k)_q$ affine source if there exist linearly independent vectors $a_1,\cdots,a_k \in \F_q^n$ and another vector $b\in\F_1^n$ s.t. $X$ is sampled by choosing $x_1,\cdots,x_k\in\F$ uniformly and independently and computing
    \begin{align*}
        X = \sum_{i=1}^k x_ia_i + b.
    \end{align*}
\end{definition}

The min-entropy of affine source coincides with its standard Shannon entropy, we simply use $H(X)$ to stand for the entropy of an affine source $X$.

The following definition is a specialization of conditional min-entropy for affine sources. It is well defined by Lemma~\ref{lemma:affine conditioning}.
\begin{definition}[Conditional min-entropy for affine sources]
Let $W$ and $Z$ be two affine sources. Define 
\begin{align*}
    H(W\mid Z) = H(W\mid_{Z=z}),\;\forall z\in \Supp(Z).
\end{align*}
\end{definition}

\begin{lemma}\label{lemma:affine bound}
Let $X, Y, Z$ be affine sources. Then
$H(X \mid (Y, Z)) \ge H(X \mid Z) - \log(\Supp(Y))$.
\end{lemma}

\begin{lemma}[Affine conditioning~\cite{Li:CCC:11}]
    \label{lemma:affine conditioning}
    
    Let $X$ be any affine source on $\{0,1\}^n$. Let $L:\{0,1\}^n \to \{0,1\}^m$ be any affine function. Then there exist independent affine sources $A,B$ such that:
    \begin{itemize}
        \item $X = A + B$
        \item There exists $c\in \{0,1\}^m$, such that for every $b\in \Supp(B)$, it holds that $L(b) = c$.
        \item $H(A) = H(L(A))$ and there exists an affine function $L^{-1}: \{0,1\}^m \to \{0,1\}^n$ such that $A = L^{-1}(L(A))$.
        \item $H(X \mid_{L(X)=\ell}) = H(B)$ for all $\ell\in \Supp(L(X))$.
    \end{itemize}
\end{lemma}

We will also need the following Lemma from~\cite{Li:CCC:11} when we do a sequential conditioning on blocks of an affine source or arguing about the total entropy of blocks of an affine source.

\begin{lemma}[Affine entropy argument~\cite{Li:CCC:11}]
    \label{lemma:affine entropy}
    Let $X$ be any affine source on $\bin^n$. Divide $X$ into $t$ arbitrary blocks $X=X_1\circ X_2 \circ \cdots \circ X_t$. Then there exists positive integers $k_1,\cdots,k_t$ such that,
    \begin{itemize}
        \item $\forall j,1\le j\le t$ and
        $\forall (x_1,\cdots,x_{j-1})\in \Supp(X_{1}, \cdots,X_{j-1})$, $H(X_{j}\mid_{X_{1}=x_{1},\cdots,X_{j-1}=x_{j-1}} )=k_{j}$.
        \item $\sum_{i=1}^t k_i = H(X)$.
    \end{itemize}
\end{lemma}

\begin{lemma}[\cite{Rao:CCC:09}]\label{lemma:affine-uniform}
    Let $\LExt:\bin^n \times \bin^d \to \bin^m$ be a $(k,\eps)$-strong linear extractor. Then for every $(n,k)$-affine source $X$,
    \[
        \Pr_{s\sim U_d}[\LExt(X,s)\text{ is uniform}]\ge 1-2\eps
    \]
\end{lemma}

\subsection{Bounded Independence}

\begin{lemma}[\cite{AlonGM03}]\label{lemma:gamma}
    If $\D$ be a $(t,\gamma)$-wise independent distribution on $\bin^n$. Then there exists a $t$-wise independent distribution that is $n^t \gamma$-close to $\D$. 
\end{lemma}

\BD[NOBF source] A source $X$ over $\{0,1\}^n$ is called a $(q,t,\gamma)$ non-oblivious bit-fixing (NOBF) source if there exists a subset $Q \subseteq [n]$ of size at most $q$ such that the joint distribution of the bits in $[n]\setminus Q$ is $(t,\gamma)$ independent. The bits in $Q$ are allowed to arbitrarily depends on the bits in $ [n]\setminus Q$.
\ED

\begin{lemma}[\cite{Braverman10,Tal:CCC:17}]\label{lemma:braverman}
Let $\D$ be any $t=t(m,d,\eps)$-wise independent distribution on $\bin^n$. Then for any circuit $\mathcal{C}\in \ac^0$ of depth $d$ and size $m$,
\begin{align*}
    \abs{\E_{x\sim U_n}[\mathcal{C}(x)]-\E_{x\sim\D}[\mathcal{C}(x)]} \le \eps
\end{align*}
where $t(m,d,\eps) = O(\log(m/\eps))^{3d+3}$.
\end{lemma}

\subsection{Influence of Variables}\label{subsec:influence}
\BD
Let $f:\bin^n \to \bin$ be any boolean function on variables $x_1,\cdots,x_n$. The influence of a set $Q \subseteq \cbra{x_1,\cdots,x_n}$ on $f$, denoted by $\I_Q(f)$, is defined to be the probability that $f$ is undetermined after fixing the variables outside $Q$ uniformly at random. Further, for any integer $q$ define $\I_q(f) = \max_{Q\subseteq\cbra{x_1,\cdots,x_n},|Q|=q}\I_Q(f)$. \\
More generally, let $\I_{Q,D}(f)$ denote the probability that $f$ is undetermined when the variables outside $Q$ are fixed by sampling from the distribution $D$. We define $\I_{Q,t}(f) = \max_{D\in D_t} \I_{Q,D}(f)$, where $D_t$ is the set of all $t$-wise independent distributions. Similarly, $\I_{Q,t,\gamma}(f) = \max_{D\in D_{t,\gamma}}\I_{Q,D}(f)$ where $D_{t,\gamma}$ is the set of all $(t,\gamma)$-wise independent distributions. Finally, for any integer $q$, define $\I_{q,t}(f) = \max_{Q\subseteq \cbra{x_1,\cdots,x_n},|Q|=q}\I_{Q,t}(f)$ and $\I_{q,t,\gamma}(f) = \max_{Q\subseteq \cbra{x_1,\cdots,x_n},|Q|=q}\I_{Q,t,\gamma}(f)$.
\ED

\subsection{Resilient Functions}
\BD
Let $f:\bin^n \to \bin$ be any boolean function on variables $x_1,\cdots,x_n$ and $q$ any integer. We say $f$ is $(q,\eps)$-resilient if $\I_q(f) \le \eps$. More generally, we say $f$ is $t$-independent $(q,\eps)$-resilient if $\I_{q,t} \le \eps$ and $f$ is $(t,\gamma)$-independent $(q,\eps)$-resilient if $\I_{q,t,\gamma}(f) \le \eps$.
\ED

The following Lemma proved in~\cite{ChattopadhyayZ:Annals:19} says that constructing extractors for $(q,t,\gamma)$-non-oblivious bit-fixing sources reduces to constructing $(t,\gamma)$-independent $(q,\eps_1)$-resilient functions.

\BL[\cite{ChattopadhyayZ:Annals:19}]
Let $f: \bin^n \to \bin$ be a boolean function that is $(t,\gamma)$-independent $(q,\eps_1)$-resilient. Further suppose that for any $(t,\gamma)$-wise independent distribution $\mathcal{D}$, $|\E_{x\sim \mathcal{D}}-\frac{1}{2}| \le \eps_2$. Then $f$ is an extractor for $(q,t,\gamma)$-non-oblivious bit-fixing sources with error $\eps_1+\eps_2$.
\EL

The following Theorem proved in~\cite{ChattopadhyayZ:Annals:19} says that upper bounding $\I_q(f)$ (or $\mathrm{bias}(f)$) translates to upper bounding $\I_{q,t,\gamma}(f)$ (or $\mathrm{bias}(f)$ under $(t,\gamma)$-wise distribution) in the case that $f$ is a constant depth monotone circuit.

\BT[\cite{ChattopadhyayZ:Annals:19}]
There exists a constant $b>0$ such that the following holds: Let $f:\bin^n \to\bin$ be a monotone circuit in $\ac^0$ of depth and size $m$ such that $\abs{\E_{x\sim U_n}[f(x)]-\frac{1}{2}}\le \eps_1$. Suppose $q>0$ is such that $\I_q(f)\le\eps_2$. If $t\ge b(\log(5m/\eps_3))^{3d+6}$, then $\I_{q,t}(f) \le \eps_2+\eps_3$ and $\I_{q,t,\gamma}(f) \le \eps_2+\eps_3+\gamma n^t$. Further, for any distribution $\mathcal{D}$ that is $(t,\gamma)$-wise independent, $\abs{\E_{x\sim \mathcal{D}}[f(x)]-\frac{1}{2}}\le \eps_1+\eps_3+\gamma n^t$.
\ET

The following gives extractors for bit-fixing sources with $k^{\Omega(1)}$ output length.
\begin{theorem}[\cite{Li:FOCS:16}]\label{thm:BFExt}
    There exists a constant $c$ such that for any constant $\delta > 0$ and all $n \in \N$, there exists an explicit extractor $\BFExt:\bin^n \to \bin^m$ such that for any $(q,t,\gamma)$ non-oblivious bit-fixing source $X$ on $n$ bits with $q \le n^{1-\delta}, t\ge c\log^{21} n$ and $\gamma \le 1/n^{t+1}$, we have that 
    \begin{align*}
        \abs{\BFExt(X) - U_m} \le \eps
    \end{align*}
    where $m = \Omega(t)$ and $\eps = n^{-\Omega(1)}$.
\end{theorem}


\cite{Viola14} showed that Majority is an extractor for $(q=n^{1/2-\alpha},t=O(1),0)$-NOBF source. Combining his result with Lemma~\ref{lemma:gamma} yields the following theorem.

\begin{lemma}[\cite{Viola14}]\label{lemma:maj}
Let $\Maj: \bin^n \to \bin$ be the majority function such that $\Maj(x) = 1 \iff \sum_i x_i \ge \lceil n/2 \rceil$. Then there exists a constant $C_{\ref{lemma:maj}}$ such that for every $(q,t,\gamma)$-NOBF source $X\in \bin^n$,
\begin{align*}
    \abs{\Maj(X)-U_1} \le C_{\ref{lemma:maj}}\pbra{\frac{\log t}{\sqrt{t}}+\frac{q}{\sqrt{n}}} + n^t \gamma
\end{align*}
\end{lemma}


\subsection{Affine Correlation Breakers}

\begin{definition}\label{def:AffineAdvCB}
    $\AffineAdvCB:\bin^n \times \bin^d \times \bin^a \to \bin^m$ is a $t$-affine correlation breaker for entropy $k$ with error $\eps$ (or a $(t,k,\eps)$-affine correlation breaker for short) if for every distributions $X,X^{[t]},A,A^{[t]},B,B^{[t]} \in \bin^n$, $Y,Y^{[t]} \in \bin^d$ and strings $\alpha,\alpha^{[t]} \in \bin^a$ such that 
    \begin{itemize}
        \item $X = A + B$, $X^i = A^i+B^i$ for every $i\in [t]$
        \item $H_\infty(A) \ge k$ and $Y$ is uniform
        \item $(A,A^{[t]})$ is independent of $(B,B^{[t]},Y,Y^{[t]})$
        \item $\forall i\in [t]$, $\alpha \neq \alpha^i$,
    \end{itemize}
    it holds that 
    \begin{align*}              (\AffineAdvCB(X,Y,\alpha)\approx_\gamma U_m)\mid (\cbra{\AffineAdvCB(X^i,Y^i,\alpha^i)}_{i\in [t]}).
    \end{align*}
    We say $\AffineAdvCB$ is strong if 
    \begin{equation}\label{eqn:acb}
        (\AffineAdvCB(X,Y,\alpha) \approx_\gamma U_m)\mid (Y,\cbra{\AffineAdvCB(X^i,Y^i,\alpha^i),Y^i}_{i\in [t]}).
    \end{equation}
\end{definition}

\begin{theorem}[\cite{Li:FOCS:23}]\label{thm:AffineAdvCB}
    For any $t$, there exists an explicit strong $t$-affine correlation breaker $\AffineAdvCB:\bin^n \times \bin^d \times \bin^a \to \bin^m$ with error $O(t\eps)$ for entropy $k = O(ta + tm + t^2\log(n/\eps))$, where $d = O(ta+tm+t\log^3(t+1)\log(n/\eps))$.
\end{theorem}

\begin{lemma}\label{lemma:acb-to-U_t}
    Let $X_1,\cdots,X_t$ be random variables, such that each $X_i$ takes values $0$ and $1$. Further suppose that for any subset $S=\{s_1,\cdots,s_r\}\subseteq [t]$,
    \[
        (X_{s_1},X_{s_2},\cdots,X_{s_r})\approx_\eps (U_1,X_{s_2},\cdots,X_{s_r}).
    \]
    Then 
    \[
        (X_1,\cdots,X_t)\approx_{t\eps} U_t
    \]
\end{lemma}

\subsection{Non-malleability and \texorpdfstring{$\oplus$}{Lg}}\label{subsec:xor-nm}
The following lemma is a special case of an extension of Vazirani's XOR Lemma. 
\begin{lemma}[\cite{DodisLWZ:SJoC:14,ChattopadhyayL:ccc:2023}]\label{lemma:nm-xor}
    Let $(W,W')$ be a random variable over $(\F_2)^2$. If $W\approx_\eps U_1$ and $(W\oplus W')\approx_\eps U_1$, then 
    \[
        (W\approx_{4\eps}U_1)\mid W'.
    \]
\end{lemma}
We can think of $W$ as the output of the extractor $f$ on the source $X$ and $W'$ as the output on the tampered source $X'$. The lemma shows that if $f(X)\oplus f(X')\approx U$, then it holds that $f(X) \approx U \mid f(X')$, which essentially coincides with the definition of the non-malleable extractor.
\section{\texorpdfstring{$\oplus$}{Lg} of Resilient Functions}
\label{sec:xor}

\subsection{\texorpdfstring{$\oplus$}{Lg} of Derandomized Ajtai-Linial Functions}

\begin{theorem}\label{thm:ac0-xor}
    There exists a constant $c$ such that for any $\delta>0$ and every large enough integer $n \in \N$,
    there exists an efficiently computable monotone boolean function $g:\bin^{2n} \to \bin^m$ that computes $\oplus$ of two copies of the boolean function $\BFExt:\bin^n \to \bin^m$ from Theorem~\ref{thm:BFExt}, i.e., $g(x):=\BFExt(x_1,\cdots,x_n)\oplus \BFExt(x_{n+1},\cdots,x_{2n})$, satisfying for any $q>0$, $t\ge c(\log n)^{21}$ and $\gamma < 1/(2n)^{t+1}$, 
    \begin{itemize}
        \item $g$ is a depth $6$ circuit of size $n^{O(1)}$;
        \item for any $(t,\gamma)$-wise independent distribution $\mathcal{D}$ on $\bin^{2n}$, $\abs{\E_{x\sim \mathcal{D}}[g(x)]-\frac{1}{2}}\le \frac{1}{n^{\Omega(1)}}$;
        \item $\mathbf{I}_{q,t,\gamma}(g) \le q/n^{1-\delta}$.
    \end{itemize}
\end{theorem}
\begin{proof} 
    First note that
    $\oplus$ of two input bits can be computed by depth-$2$ $\ac^0$ circuit of size $2^{O(2)} = O(1)$. However, since each output bit of $\BFExt$ is already an $\oplus$ of multiple output bits from $\ac^0$ circuits, this top layer $\oplus$ effectively collapses with the topmost $\oplus$ layer of $\BFExt$. Therefore, $\oplus$ of two functions in Theorem~\ref{thm:BFExt} is still a depth-$6$ $\ac^0$ circuit with a larger size which is still bounded by $n^{O(1)}$. 
    Moreover, the function $g$ is monotone. Recall that the resilient function constructed in~\cite{ChattopadhyayZ:Annals:19} is a depth-$4$ monotone $\ac^0$ circuit. Specifically, it is $\mathsf{AND}$ of the Tribes function whose variables are replaced by CNF formulas. The circuit is monotone because the bottom layers are monotone $\mathrm{CNF}$. As $g$ only superimposes $\oplus$ atop $f$ without modifications to the lower layers, $g$ retains its monotonicity. The balanced Condition can be derived from Lemma~\ref{lemma:braverman}, which says fooling the depth $6$ $\ac^0$ circuit needs $t$ to be $O(\log(2n))^{3\cdot 6 + 3}=O(\log^{21}(n))$. By Lemma~\ref{lemma:gamma}, the error due to $\gamma$ is bounded by $(2n)^t\cdot \gamma \le 1/(2n)$. 
    Lastly, by Theorem~\ref{thm:BFExt}, $\I_{q,t,\gamma}(f)\le n^{-\frac{\delta}{8}}$, then by a union bound $\I_{q,t,\gamma}(g)\le 2 \cdot n^{\frac{\delta}{8}} = n^{-\Omega(1)}$.
\end{proof}

\subsection{\texorpdfstring{$\oplus$}{Lg} of Majority Functions}

\begin{theorem}\label{thm:maj-xor}
Let $\Maj: \bin^n \to \bin$ be the majority function from Lemma~\ref{lemma:maj}. Let $X,Y\in \bin^n$ any two random variables such that the joint distribution $(X,Y)\in \bin^{2n}$ is a $(q,t,\gamma)$-NOBF source. Then there exists a constant $C_{\ref{thm:maj-xor}}$ such that the following holds,
\begin{align*}
    \abs{\Maj(X)\oplus \Maj(Y)-U_1} \le C_{\ref{thm:maj-xor}}\pbra{\frac{1}{\sqrt{t}}+\frac{q}{\sqrt{n}}} + (2n)^t \gamma
\end{align*}
\end{theorem}

\begin{proof}
We first recall a Theorem on bounded independence fooling functions of halfspaces.
\begin{theorem}[\cite{POWZuckerman:CCC:10}]
    Suppose $f$ is computable as a size-$s$, depth-$d$ function of halfspaces over independent random variables $x_1,\cdots,x_n$. If we assume the $x_j$'s are discrete, then $k$-wise independent suffices to $\eps$-fool $f$, where 
    \[ k = \widetilde{O}(d^4s^2/\eps^2)\cdot \poly(1/\alpha). \]
    Here $0<\alpha \le 1$ is the least nonzero probability of any outcome for an $x_j$. Moreover, the same result holds with $\alpha=1$ for certain continuous random variables $x_j$, including Gaussians (possibly of different variance) and random variables that are uniform on (possibly different) intervals.
\end{theorem}
Given a $t$-wise independent distribution on $2n$ bits. Divide the bits into $2$ blocks of equal length. Applying an $\Maj:\bin^n\to \bin$ on each block. Then taking $\oplus$ of the two outputs is $\eps$-close to uniform as long as $t=\widetilde{O}({2}^4\cdot 2^{4}/\eps^2)$.

Now, instead of considering $t$-wise independent distribution, we consider $(2q,t,\gamma)$-wise independent distribution on $2n$ bits.

Let $\D$ be a $(2q,2n,0)$ distribution on $2n$ bits, then by property of binomial distribution
\begin{equation}\label{eqn:resilience}
    f(U_{2n}) \approx_{O(\frac{q}{\sqrt{n}})} f(\D). 
\end{equation}

Since by Lemma~\ref{lemma:gamma}, $(X,Y)$ is $(2n)^t\gamma$ close to a $(2q,t,0)$-NOBF source on $2n$ bits when $X$ and $Y$ are each instead from a $(q,t,\gamma)$ source, combining  with Eqn.~\eqref{eqn:resilience} and the fact that $\oplus$ of two $\Maj$ functions is close to uniform, we have

\begin{align*}
\abs{f(\mathcal{D}) - U_1} \le O\pbra{\frac{1}{\sqrt{t}}+\frac{q}{\sqrt{n}}}+(2n)^t\gamma.
\end{align*}
\end{proof}

\section{Advice Generators}
\label{sec:advGen}
In this section, we present advice generators that generate advice with very low input entropy.
Due to the limit we are pushing the non-malleable extractors, the advice generators need to be very particular about parameters. It turns out that our advice generators need to achieve simultaneously entropy $k=O(\log(1/\eps))$, output length $a=O(\log(1/\eps))$ where $\eps$ is the error.

\subsection{Advice Generators for Two Independent Sources}

Algorithm~\ref{alg:advGen} takes as input two independent sources $X$ and $Y$ of $O(\log(1/\eps))$ entropy together with a short uniform string that is a deterministic function of $X$ and $Y$ and is uniform. We assume that the part from $X$ has been conditioned on.
\begin{algorithm}[H]
    \caption{$\advGen(x,y,y_1)$}
    \label{alg:advGen}
    \begin{algorithmic}
        \medskip
        \State \textbf{Input:} Bit strings $x,y \in \bin^n$, a uniform string $y_1\in \bin^d$ that is a deterministic function of $y$ and $x$ where the part from $x$ has already been conditioned upon and an error parameter $\eps$.
        \State \textbf{Output:} Bit string $z$ of length $m_1+d+2D' = O(\log(1/\eps))$. 
        \State \textbf{Subroutines and Parameters:} \\
        Let $\Ext:\bin^n \times \bin^d \to \bin^{m_1}$ be the strong seeded extractor from Theorem~\ref{thm:GUV09} with output length $m_1$ to be decided later, entropy $k = 2m_1$, seed length $d=O(\log(n/\eps_1))$, and
        error $\eps_1 = \frac{\eps}{2}$. \\
        Let $\Enc:\bin^n \to \bin^{n_1}$ be the encoding function of an asymptotically good linear binary code with constant relative rate $1/\lambda$ and constant relative distance $\beta$. Thus $n_1 = \lambda n$.\\
        Let $\Samp:\bin^{m_1} \to [n_1]^{D'}$ be a $(\beta,\beta/2,\eps_2)$ averaging sampler where $\eps_2\le \frac{\eps^2}{4}$ from Theorem~\ref{thm:avg-samp} with $D'=O(\frac{4}{\beta^2}\log(1/\eps_2))$ and $m_1=\log(\frac{\beta^2}{4}\lambda n/\log(1/\eps_2))+\log(1/\eps_2)\poly(2/\beta)+O(1)$.
        \\\hrulefill \\
        \begin{enumerate}
            \item Let $x_1 = \Ext(x,y_1)$.
            \item Let $w_1 = \Enc(y)_{\mid\Samp(x_1)}$, $w_2 = \Enc(x)_{\mid\Samp(\Slice(y_1,m_1))}$. 
            \item Output $z = x_1 \circ y_1 \circ w_1 \circ w_2$.
        \end{enumerate} 
    \end{algorithmic}
\end{algorithm}

\begin{theorem}\label{thm:advGen}
For any $\eps \ge n^{-\Omega(1)}$, let $X,Y$ be two independent sources where $\minH(X) \ge O(\log(1/\eps))$, $Y_1$ be a string of length $d=O(\log(1/\eps))$ that is a deterministic function of $Y$ and $X$, and is uniform and a deterministic function of $Y$ and $X$ with the part from $X$ already been conditioned on. Then for any tampering $(X',Y'):=(f(X),g(Y))$ of $(X,Y)$ s.t. at least one of $f$ and $g$ has no fixed point, then it holds that $\Pr[\advGen(X,Y,Y_1) \neq \advGen(X',Y',Y'_1)] \ge 1 -\eps$. 
\end{theorem}
\begin{proof}
\hfill\break
\begin{enumerate}
    \item Since $Y_1$ is uniform and independent of $X$, by Theorem~\ref{thm:GUV09}, it holds that $(X_1\approx_{\eps_1}U_{m_1}) \mid Y_1$. In the later analysis, we can first assume that $X_1$ is uniform and add back the $\eps_1$ error in the end.
    \item 
    If $H_\infty(X_1 \mid X_1=X_1') \le m_1-\log(1/\eps_1)$ or $H_\infty(Y_1 \mid Y_1=Y_1') \le m_1-\log(1/\eps)$, we claim that 
    $\Pr[X_1\circ Y_1 \neq X_1'\circ Y_1'] \ge 1-\eps$.
    
    If $H_\infty(X_1 \mid X_1=X_1') \le m_1-\log(1/\eps_1)$, then
    \begin{equation}\label{eqn:two-min}
        \exists x\in \Supp(X_1\mid X_1=X_1'),\;\text{s.t. }\Pr[X_1=x\mid X_1=X_1']\ge \frac{1}{\eps_1 \cdot 2^{m_1}}.
    \end{equation}
    By Bayes' rule, 
    \begin{equation}\label{eqn:two-bayes}
        \begin{split}
            \Pr[X_1=x \mid X_1=X_1'] 
        & = \frac{\Pr[X_1=X_1'\mid X_1=x]\cdot\Pr[X_1=x]}{\Pr[X_1=X_1']} \\
        & \le \frac{\Pr[X_1=x]}{\Pr[X_1=X_1']} \\
        &\le \frac{1}{2^{m_1}} \cdot \frac{1}{1-\Pr[X_1\neq X_1']}.
        \end{split}
    \end{equation}
    Combining Inequalities~\eqref{eqn:two-min} and~\eqref{eqn:two-bayes}, we have
    \[
        \frac{1}{\eps_1 \cdot 2^{m_1}} \le \frac{1}{2^{m_1}} \cdot \frac{1}{1-\Pr[X_1\neq X_1']} \iff \Pr[X_1\neq X_1'] \ge 1 - \eps_1.
    \]
    Adding back the $\eps_1$ error, it already holds that
    $\Pr[X_1\circ Y_1 \neq X_1' \circ Y_1']\ge 1-\eps_1-\eps_1= 1-\eps$, i,e., the advice differ with the desired probability. The case where $H_\infty(Y_1 \mid Y_1=Y_1') \le m_1-\log(1/\eps)$ can be handled similarly.
    
    Hence, afterward in the proof, we focus on the subsource $(X_1 \mid  (X_1'=X_1))$ and $(Y_1 \mid  (Y_1'=Y_1))$ and
    assume that $H_\infty(X_1 \mid  (X_1'=X_1))\ge m_1-\log(1/\eps_1)$ and $H_\infty(Y_1 \mid  (Y_1'=Y_1))\ge m_1-\log(1/\eps)$. 
    \item If $f$ is without fixed points, then we condition on $X$, note that $f(X)\neq X$ is also fixed. Moreover, it holds that $\Enc(X \oplus f(X))$ has at least $\beta n_1$ non-zero positions.
    \item It holds with probability $1-\eps_2/\eps$ (since $Y_1$ loses $\log(1/\eps)$ entropy as we condition on $Y_1= Y_1'$, and this cause a $2^{\log(1/\eps)}=1/\eps$ multiplicative factor to the error of $\Samp$ by Remark~\ref{rmk:samp-deficient}) that $\Samp(Y_1)$ intersects more than $\beta n_1/2$ positions.
    \item When $g$ has no fixed point, conditioned on $Y$ (and thus $Y'$, $Y_1$ and $Y_1'$), we have $\Enc(Y\oplus Y')$ has at least $\beta n_1$ non-zero positions and $H_\infty(X_1) \ge m_1 - \log(1/\eps_1)$. Therefore, with probability $1-\eps_2/\eps_1$, it holds that $\Samp(X_1)$ intersects more than $\beta n_1/2$ positions.
    \item Overall, with probability at least $1-\eps_1-\eps_2/\eps_1 \ge 1-\eps$, it holds that $\advGen(X,Y,Y_1) \neq \advGen(X',Y',Y_1')$.
\end{enumerate}
\end{proof}

\subsection{Advice Generators for Affine Source}

We adopt a construction roughly the same as~\cite{ChattopadhyayL:STOC:17}. 

Algorithm~\ref{alg:affadvGen} takes as input an affine source $X$ of $O(\log(1/\eps))$ entropy and a uniform string $Y$ of $O(\log(1/\eps))$ length that is a linear function $L$ of $X$ and generates advice of $O(\log(1/\eps))$ length. The guarantee is that for any affine tampered version $X'$ of $X$, the advice generated on $(X,L(X))$ differs from $(X',L(X'))$ with probability $1-\eps$. 

To achieve this, we apply pseudorandom objects with suitable parameters so that the $\AffineAdvGen$ possesses the parameters that fit our needs. Specifically, the advice generator uses the optimal averaging sampler from Theorem~\ref{thm:avg-samp} to achieve potentially even logarithmic advice length; to generate advice on affine sources of logarithmic entropy with logarithmic seed length, we make use of the strong linear seeded extractor from Lemma~\ref{lemma:log_n_lext} by~\cite{ChattopadhyayGL21}.

The advice generated has the property that, conditioned on a substring of the advice, which is a linear function of the affine source, the remaining bits of the advice are also linear functions of the affine source.

\begin{algorithm}[H]
    \caption{$\AffineAdvGen(x,y)$}
    \label{alg:affadvGen}
    \begin{algorithmic}
        \medskip
        \State \textbf{Input:} Bit strings $x \in \bin^n$ and $y \in \bin^d$ where $d = O(\log(n))$; an error parameter $\eps$.
        \State \textbf{Output:} Bit string $z$ of length $\ell_1+\ell_2+D+n_3 = O(\log(n))$.
        \State \textbf{Subroutines and Parameters:} \\
        Let $\AExt_1:\bin^{\ell_1} \to \bin^{m_1}$ from Theorem~\ref{thm:BAExt} with input length $\ell_1=O(\log(1/\eps_1))$, entropy $k_1 = \Omega(\ell_1)$, output length $m_1 = \Omega(\ell_1)$ and
        error $\eps_1\le \eps/4$. \\
        Let $\AExt_2:\bin^{\ell_2} \to \bin^{m_2}$ from Theorem~\ref{thm:BAExt} with input length $\ell_2 = O(\log(1/\eps_2))$, entropy $k_2 = \Omega(\ell_2)$, output length $m_2 = \Omega(\ell_2)$ and
        error $\eps_2\le \eps/4$. \\
        Let $\Enc:\bin^n \to \bin^{n_1}$ be the encoding function of an asymptotically good linear binary code with constant relative rate $1/\lambda$ and constant relative distance $\beta$. Thus $n_1 = \lambda n$.\\
        Let $\Samp:\bin^{m_1} \to [n_1]^{D}$ be a $(\beta,\beta/2,\eps_3)$ averaging sampler where $\eps_3\le\eps/4$ from Theorem~\ref{thm:avg-samp} with $D=O(\frac{4}{\beta^2}\log(1/\eps_3))$ and $m_1=\log(\frac{\beta^2}{4}\lambda n/\log(1/\eps_3))+\log(1/\eps_3)\poly(2/\beta)$.\\
        Let $\LExt:\bin^n\times\bin^{m_2}\to \bin^{n_3}$ be the linear strong seeded extractor from Lemma~\ref{lemma:log_n_lext} with output length $n_3=\log(n)$, entropy $k_3=C_{\ref{lemma:log_n_lext}}(n_3+\log(1/\eps_4))=k/2$, seed length $m_2=O(n_3/2+\log(n/\eps_4)+\log(n_3/\eps_4))$ and error $\eps_4 \le \eps/8$. 
        \\\hrulefill \\
        \begin{enumerate}
            \item Let $y = y_1 \circ y_2 \circ y_3$, where $|y_1| = \ell_1$, $|y_2| = \ell_2$. 
            \item Let $s_1 = \AExt_1(y_1)$, $s_2 = \AExt_2(y_2)$.
            \item Let $w_1 = \Enc(x)_{\mid\Samp(s_1)}$, $w_2 = \LExt(x,s_2)$. 
            \item Output $z = y_1 \circ y_2 \circ w_1 \circ w_2$.
        \end{enumerate} 
    \end{algorithmic}
\end{algorithm}

\begin{theorem}\label{thm:AffineAdvGen}
    For any $\eps$, let $X$ be an affine source of entropy $k = O(\log(1/\eps))$, $L:\bin^n \to \bin^d$ be an affine function with $d\le n$ such that $Y = L(X)$ is uniform. There exists a function $\AffineAdvGen:\bin^n \times \bin^d \to \bin^a$ such that for any affine tampering $X':=\mathcal{A}(X)$ of $X$ without fixed points and $Y'=L(X')$, it holds that $\Pr[\AffineAdvGen(X,Y) \neq \AffineAdvGen(X',Y')] \ge 1 -\eps$, where $a=O(\log(1/\eps))$.
\end{theorem}
\begin{proof}
    If $H(Y_1\mid (Y_1 = Y_1'))\le \ell_1-\log(1/\eps)$ or $H(Y_2\mid (Y_2=Y_2'))\le \ell_2-\log(1/\eps)$, it already holds that $\Pr[Y_1\circ Y_2 \neq Y_1' \circ Y_2'] \ge 1-\eps$, i.e., the advice differ with the desired probability. Therefore, afterward in the proof, we focus on the subsource $(Y_1 \mid  (Y_1'=Y_1))$ and $(Y_2 \mid  (Y_2'=Y_2))$ and
    assume that $H(Y_1 \mid  (Y_1'=Y_1))\ge \ell_1-\log(1/\eps)$ and $H(Y_2 \mid  (Y_2'=Y_2))\ge \ell_2-\log(1/\eps)$. 
    
    We want to show that either $W_1 \oplus W_1' = \Enc(X \oplus X')_{\mid\Samp(S_1)}\neq 0$ or $W_2 \oplus W_2' = \LExt(X\oplus X',S_2)\neq 0$. 
    
    Note that since $Y$ is a linear function of $X$, and both $Y_1$ and $Y_2$ are linear functions of $Y$, it holds that $Y_1$ and $Y_2$ are linear functions of $X$. Since $Y$ is uniform, it holds that $Y_1$ and $Y_2$ are independent.
    By Lemma~\ref{lemma:affine conditioning}, $X$ can be written as a sum of two affine sources $A$ and $B$ such that $Y_2(X) = Y_2(A)$ and $Y_2(B)=0$; $X'$ can be written as a sum of two affine sources $A'$ and $B'$ such that $Y_2(X')=A'$ and $Y_2(X')=B'$.
    Moreover, $Y_1$ is a deterministic function only of $B$. 
    \begin{description}
        \item [Case 1.] $H(X\oplus X'\mid Y_2)=H((\mathcal I \oplus \mathcal A)(X)\mid Y_2)\le k_3$.
        
         Fix $Y_2$, and condition on this fixing. Fix $X\oplus X'$, which is a linear function of $X$. Fixing $X\oplus X'$ will cause $\le H((\mathcal I \oplus \mathcal A)(B))= H(X\oplus X')$ entropy loss to $B$ and therefore at most $H(X\oplus X')$ entropy loss to $Y_1$. Since $H(Y_1\mid (Y_1 = Y_1',\; Y_2,\;X\oplus X')) \ge \ell_1-\log(1/\eps)-k_3 \ge k_1$, it holds that $Y_1$ still has suffient entropy after this fixing, and therefore $S_1\approx_{\eps_1} U_{m_1}$. Now, since $X\neq X'$ and $X\oplus X'$ is fixed, it holds that $\Enc(X\oplus X')$ is a fixed string with at least $\beta\lambda n$ non-zero bits.
    By the property of the sampler, $\Samp(S_1)$ intersects one of these bits with probability at least $1-\eps_3-\eps_1$. And if this happens it holds that $W_1\neq W_1'$, and thus $Z\neq Z'$.
    \item [Case 2.] $H(X\oplus X'\mid Y_2)=H((\mathcal I \oplus \mathcal A)(X)\mid Y_2)\ge k_3$. 
    
    Note that this is equivalent to $H(B\oplus B')=H((\mathcal I \oplus \mathcal A)(B) )\ge k_3$. Since $H(Y_2\mid (Y_2'=Y_2))\ge\ell_2-\log(1/\eps)\ge k_2$, by Theorem~\ref{thm:BAExt}, $S_2 \approx_{\eps_2} U_{m_1}$. Now, condition on $Y_2$, then by Lemma~\ref{lemma:affine-uniform} we have with probability $1-\eps_2-2\eps_4$, it holds that $W_2\oplus W_2'= U_{n_3}$.
    \end{description}
    Therefore with probability $1-\eps_1-\eps_2-\eps_3-2\eps_4\ge 1-\eps$, $Z\neq Z'$.
\end{proof}
\section{Non-Malleable Extractors for Polylogarithmic Entropy}
\label{sec:nmext-polylog}

In this section, we present two-source and affine non-malleable extractors that works for $\polylog(n)$ entropy.

\subsection{Two-Source Non-Malleable Extractors}\label{sec:nm2ext}

\begin{algorithm}[H]
    \caption{$\tnmext(x,y)$ for $\polylog(n)$ entropy}
    \label{alg:tnmExt-pe}
    \begin{algorithmic}
        \medskip
        \State \textbf{Input:} $x,y \in \bin^n$ --- two $n$ bit strings.
        \State \textbf{Output:} $w \in \bin^m$ --- a bit string with length $m = \polylog(n)$.
        \State \textbf{Subroutines and Parameters:} \\
        Let $\Ext: \bin^n \times \bin^{d'} \to \bin^{d_1}$ be the strong seeded extractor from Theorem~\ref{thm:GUV09} with entropy $k_1 = 2d_1$, seed length $d' = c_{\ref{thm:GUV09}}\log(n/\eps_1)$, output length $d_1$ to be defined later and error $\eps_1 = 2^{-4(a+d)t}\cdot n^{-4}$. Let $D' = 2^{d'},D_1 = 2^{d_1}$. \\
        Let $\Ext':\bin^n \times \bin^d \to \bin^{d'}$ be a seeded extractor of min-entropy $k'=O(\log(n))$, seed length $d=O(\log(n))\ge 2\log(1/\eps')$, error $\eps'=1/n$ and output length $d'$. Note that $\Ext'$ is also a $(\eps',2^{-k'})$-sampler for entropy $2k'$ by Theorem~\ref{thm:samp}.\\
        Let $\advGen:\bin^n\times \bin^{n}\times \bin^{d_2} \to \bin^a$ be the advice generator from Theorem~\ref{thm:advGen} with entropy $k_2=O(\log(1/\eps_2))$, $d_2=O(\log(n/\eps_2))$, advice length $a=O(\log(1/\eps_2))$ and error $\eps_2 = 2^{-4dt}\cdot n^{-2}$.\\
        Let $\AffineAdvCB:\bin^n\times \bin^{d_1} \times \bin^{a+d}\to \bin$ be the $t$-affine correlation breaker from Theorem~\ref{thm:AffineAdvCB} with entropy $k_3=O(ta+td+t+t^2\log(n/\eps_3))$, seed length $d_1=O(ta+td+t+t\log^3(t+1)\log(n/\eps_3))$ where $\eps_3=2^{-2(a+d)t}\cdot n^{-2}$.\\
        Let $\BFExt:\bin^D \to \bin^m$ be the extractor for NOBF source from Theorem~\ref{thm:BFExt}. \\
        Let $\LExt':\bin^n \times \bin^m \to \bin^{m'}$ be a $(\Omega(k),n^{-\Omega(1)})$-strong linear seeded extractor from Theorem~\ref{thm:LExt-Trev} where $m\ge O(\log^2(n))$ and $m' = \Omega(k)$.
        \\\hrulefill \\ 
        \begin{enumerate}
            \item Let $r_i = \Ext(y,\Ext'(x,i))$ for every $i\in [D]$.
            \item Let $\alpha_i = \advGen(x,y,\Slice(r_i,d_2))\circ i$.
            \item Let $z_i = \AffineAdvCB(x,r_i,\alpha_i)$ for every $i\in [D]$.
            \item Let $w=\BFExt(z_1,\cdots,z_D)$.
            \item Output $v = \LExt'(y,w)$.
        \end{enumerate} 
    \end{algorithmic}
\end{algorithm}

\begin{theorem}\label{thm:2nmext-pe}
    There exists an explicit construction of a two-source non-malleable extractor for entropy $k\ge \polylog(n)$ with output length $\Omega(k)$ and error $n^{-\Omega(1)}$.
\end{theorem}

\begin{proof}
\hfill\break
Let $X':=f(X)$, $Y':=g(Y)$. Let $\eps = n^{-\Omega(1)}$ and $\delta = \frac{1}{(2D)^ttn}$ be two error parameters that will be used in the analysis below.
Let $k_x = \max\{k_2+d't+\log(1/\delta),k_3+d't+at+\log(1/\delta),2k'\}$ and $k_y=k_1$, where $k_x$ and $k_y$ are the minimum amount of entropy that $X,X'$ and $Y,Y'$ need to possess respectively so that $\tnmext$ works as a two-source non-malleable extractor on $(X,Y,X',Y')$ and a two-source extractor on each of $(X,Y)$ and $(X',Y')$. \\
Our proof trifurcates depending on $H_\infty(X')$ and $H_\infty(Y')$. 
Let $\minH(X)\ge 2k_x+\log(1/\eps)$ and $\minH(Y) \ge 2k_y+\log(1/\eps)$. 

\begin{description}
    \item[Case 1.] $\minH(Y') < k_y$. 
    
    Without loss of generality, assume that $Y$ is a flat source. Let $\BAD$ be defined as $\BAD:=\{y'\in g(Y):|g^{-1}(y')|\le 2^{k_y}\}$. 
\begin{itemize}
    \item If $\Pr[y'\in \BAD]\le \eps$, our analysis can be confined to the subsources \(Y\mid \overline{\BAD}\). Note that we use the notation \(\overline{\BAD}\) to represent the event \(y'\in \Supp(Y)\setminus \BAD\). Given that these subsources contribute to a density of \(1-\eps\), neglecting \(Y\mid \BAD\) results in only an \(\eps\) deviation in the final error.
    \item Otherwise, \(\Pr[y'\in\BAD] > \eps\), then we can infer that 
    \[\max_{y'\in\BAD}\Pr[Y'=y'\mid Y'\in \BAD]\le \frac{2^{k_y}}{\eps|\Supp(Y)|}.\] From this, we deduce the following inequalities: 
    \(\minH(Y' \mid Y'\in \BAD) \ge \minH(Y)-k_y-\log(1/\eps) \ge k_y\) and \(\minH(Y \mid Y'\in \BAD)\ge -\log\pbra{\frac{1}{\eps|\Supp(Y)|}} \ge \minH(Y)-\log(1/\eps) \ge k_y\). We can then think of $(Y \mid Y'\in\BAD)$ and $(Y'\mid Y'\in \BAD)$ as a new pair of source and its tampering by the function $g$ restricted to $(Y\mid Y'\in \BAD)$ such that both $(Y\mid Y'\in \BAD)$ and $(Y'\mid Y'\in \BAD)$ meets the minimum entropy requirement for $\tnmext$ for the second source. And the rest of the analysis only depends on the entropy of the first source and its tampering by $f$. However, in any case, we will be able to show that $\tnmext$ runs correctly on the subsources $(Y\mid Y'\in \BAD)$ and $(Y'\mid Y'\in \BAD)$.
\end{itemize}
Now consider $y'\in\overline{\BAD}$. It holds that $\minH(Y\mid Y'=y')\ge k_y$, meaning that the entropy of $Y$ conditioned on $(Y'=y')$ still meets the requirement to carry out all the functionality of the extractors in the rest of steps, we can safely condition $Y$ on $Y'$. Now we additionally fix $W'=\tnmext(X',Y')$. From this point, we could just focus on the running of $\tnmext$ on $(X\mid W)$, $(Y \mid Y')$. And by the above bounds on the average-case entropy of the above two sub-sources, since $\tnmext$ itself resembles an algorithm for $\mathsf{2Ext}$, the correctness follows. Now, since $Y$ is just a convex combination of subsources $(Y\mid Y'\in \BAD)$ and $\{(Y\mid Y'=y')\}_{y'\in \overline{\BAD}}$, it holds that $\tnmext$ runs correctly on $(X,Y)$ and $(X',Y')$.\\
From now on, we assume that $\minH(Y')\ge k_y$. We need to consider two cases based on $\minH(X')$.
    \item[Case 2.] $\minH(X')<k_x$.
    
    Without loss of generality, assume that \(X\) is a flat source. Let \(\BAD\) be defined as \(\BAD:=\{x'\in f(X):\abs{f^{-1}(x')} \le 2^{k_x}\}\). 
    \begin{itemize}
        \item If \(\Pr[x'\in\BAD] \le \eps\) our analysis can be confined to the subsources \(X\mid \overline{\BAD}\). Note that we use the notation \(\overline{\BAD}\) to represent the event \(x'\in \Supp(X)\setminus \BAD\). Given that these subsources contribute to a density of \(1-\eps\), neglecting \(X\mid \BAD\) results in only an \(\eps\) deviation in the final error.
        \item Otherwise, \(\Pr[x'\in\BAD] > \eps\), then we can infer that \[\max_{x\in\BAD}\Pr[X'=x'\mid X'\in \BAD]\le \frac{2^k}{\eps|\Supp(X)|}.\] From this, we deduce the following inequalities: \(\minH(X'\mid X'\in \BAD) \ge \minH(X)-k_x-\log(1/\eps) \ge k_x\)  and \(\minH(X \mid X'\in\BAD)\ge -\log\pbra{\frac{1}{\eps|\Supp(X)|}} \ge \minH(X)-\log(1/\eps) \ge k_x\). Then by what we will show in Case 3 in the following paragraphs, it holds that 
        \[
            ((\tnmext(X\mid X'\in\BAD,Y)\approx_{\eps+n^{-\Omega(1)}}U_m) \mid \tnmext(X'\mid X'\in\BAD,g(Y))
        \]
    \end{itemize}

    Now consider $x'\in\overline{\BAD}$. It holds that $\minH(X\mid X'=x')\ge k_x$, meaning that the entropy of $X$ conditioned on $(X'=x')$ still meets the requirement to carry out all the functionality of the extractors in the rest of steps, we can safely condition $X$ on $X'$. Now we additionally fix $\tnmext(X',Y')$, and the rest of the analysis is similar to the case where $\minH(Y')< k_y$. 
    \item[Case 3.] $H_\infty(X')\ge k_x$, $H_\infty(Y')\ge k_y$. 
     
    \begin{lemma}\label{lemma:use-x-samp-y}
        With probability $1-2^{-k'+1}$ over the fixing of $X$, no less than $(1-\sqrt{\eps_1}-\eps')D$ indices $i\in [D]$ are good, i.e., $R_i\approx_{\sqrt{\eps_1}}U_{d_1}$; no less than $(1-\sqrt{\eps_1}-\eps')D$ indices $i\in [D]$ are good, i.e., $R'_i\approx_{\sqrt{\eps_1}}U_{d_1}$.
    \end{lemma}
    \begin{proof}
        Let $V$ be a uniform distribution on $\bin^{d'}$, it holds by Theorem~\ref{thm:GUV09} that $|(\Ext(Y,V),V)-(U_{d_1},V)|\le \eps_1$. Let $\BAD_1\subseteq [D]$ s.t. $\forall i\in \BAD_1$, it holds that $|(\Ext(Y,i),i)-(U_{d_1},i)|\ge \sqrt{\eps_1}$. Then it holds that $|\BAD_1|\le \sqrt{\eps_1}D$ since $\frac{1}{D}\sum_{i\in \BAD_1}|(\Ext(Y,i),i)-(U_{d_1},i)|\le |(\Ext(Y,V),V)-(U_{d_1},V)|\le \eps_1$. Therefore, $\forall i\in [D]\setminus \BAD_1$, it holds that $\Ext(Y,i)\approx_{\sqrt{\eps_1}}U_{d_1}$; similarly, $\exists \BAD_2\subset [D],\;|\BAD_2|\le \sqrt{\eps_1}D$ s.t. $\forall s\in [D]\setminus \BAD_2$, it holds that $\Ext(Y',s)\approx_{\sqrt{\eps_1}}U_{d_1}$.
        
        Now note that by Theorem~\ref{thm:samp}, $\Ext'$ is also a sampler. With probability $1-2^{-k'}$ over the fixings of $X$, no less than $(1-\sqrt{\eps_1}-\eps')D$ $i$'s in $[D]$ are good, i.e. $R_i\approx_{\sqrt{\eps_1}}U_{d_1}$; with probability $1-2^{-k'}$ over the fixings of $X'$ (thus also $X$, since $X'$ is a deterministic function of $X$), no less than $(1-\sqrt{\eps_1}-\eps')D$ $i$'s in $[D]$ are good, i.e. $R'_i\approx_{\sqrt{\eps_1}}U_{d_1}$. By a union bound, the statement of Lemma~\ref{lemma:use-x-samp-y} holds.
    \end{proof}
    Let $R=R_1\circ \cdots \circ R_D$, $R'=R_1'\circ\cdots \circ R_D'$, $Z= Z_1\circ \cdots \circ Z_D$, $Z' = Z_1' \circ \cdots \circ Z_D'$.
    \begin{lemma}
        Conditioned on the event that no less than $(1-\sqrt{\eps_1}-\eps')D$ rows in $R$ and no less than $(1-\sqrt{\eps_1}-\eps')D$ rows in $R'$ are close to uniform. It holds that with probability $1-(2D)^tt(\sqrt{\eps_1}+\eps_2)$ over the fixing of $X$, $(Z_1,\cdots,Z_D,Z'_1,\cdots,Z'_D)$ is a $(2(\sqrt{\eps_1}+\eps')D,t,2^{at}t(\sqrt{\eps_1}+\eps_3))\text{-}\NOBF$ source.
    \end{lemma}
    \begin{proof}
        Define a ``bad" $x\in \Supp(X)$ to be a string that satisfies the following two properties:
        \begin{enumerate}
            \item [(a)] Conditioned on the fixing of $X=x$, there exists a large set $G_1\subseteq [D],|G_1|\ge(1-\sqrt{\eps_1}-\eps')D$; $G_2\subseteq [D],|G_2|\ge (1-\sqrt{\eps_1}-\eps')D$ s.t. each row of $R$ with index in $G_1$ and each row of $R'$ with index in $G_2$ is $\sqrt{\eps_1}$-close to uniform.
            \item [(b)] Conditioned on the fixing of $X=x$, there exists two subsets $T_1\subset G_1$, $T_2\subset G_2$ with $|T_1|+|T_2|=t$ such that the concatenation of rows in $Z$ with indices in $T_1$ and rows in $Z'$ with indices in $T_2$, which we denote by $Z_{T_1,T_2}$, is $\gamma$-far from uniform, where $\gamma$ is an error parameter to be decided later.
        \end{enumerate}
        Note that if we subtract the density of ``bad" $x\in \Supp(X)$ from the density of $x\in \Supp(X)$ s.t. no less than $(1-\sqrt{\eps_1}-\eps')D$ rows in $R$ and no less than $(1-\sqrt{\eps_1}-\eps')D$ rows in $R'$ are close to uniform, we obtain the density of $x\in \Supp(X)$ such that any $t$ bits from $Z$ and $Z'$ except for $(\sqrt{\eps_1}-\eps')D$ bits in $Z$ and $(\sqrt{\eps_1}-\eps')D$ bits in $Z'$ are $\gamma$-close to uniform.
        
        Now for each $T_1,T_2\subset[D]$ with $|T_1|+|T_2|=t$, we define an event $\BAD_{T_1,T_2}$ to be the set of $x$'s in $\Supp(X)$ that satisfies the following two properties:
        \begin{enumerate}
            \item [(c)] Conditioned on the fixing of $X=x$,  each row of $R$ with index in $T_1$ and each row of $R'$ with index in $T_2$ is $\sqrt{\eps_1}$-close to uniform.
            \item [(d)] 
            Conditioned on the fixing of $X=x$, $Z_{T_1,T_2}$ is $\gamma$-far from uniform.
        \end{enumerate}
    Observe that the property (c) is determined by $t$ random variables $\{S_i=\Ext'(X,i)\}_{i\in T_1}$ and $\{S_i'=\Ext'(X',i)\}_{i \in T_2}$. Let $S$ be the concatenation of these random variables, and define $A_{T_1,T_2}$ to be the set of $s$'s in $\Supp(S)$ that makes property (c) satisfied, then we have $\Pr[\BAD_{T_1,T_2}]=\sum_{s\in A_{T_1,T_2}}\Pr[S=s]\Pr[\BAD_{T_1,T_2}\mid S=s]$.

    Now since the size of $S$ is small: $td'$, $\avgH(X\mid S=s)\ge k_2+\log(1/\delta)$, $\avgH(X' \mid S=s) \ge k_2+\log(1/\delta)$.

    Define the random variables $\Alpha_i = \advGen(X,Y,R_i)\circ i$; $\Alpha'_i = \advGen(X',Y',R_i')\circ i$. It is clear that $\Alpha_i\neq \Alpha'_j,\forall i,j\in [D], i\neq j$. Let $\Alpha$ be the concatenation of $\Alpha_i = \advGen(X,Y,R_i)\circ i$ for $i\in T_1$ and $\Alpha'_j = \advGen(X',Y',R_j')\circ j$ for $j\in T_2$. By Theorem~\ref{thm:advGen}, with probability at least $1-2(\sqrt{\eps_1}+\eps_2+\delta)|T_1\cap T_2|=1-t(\sqrt{\eps_1}+\eps_2+\delta)$ over the fixing of $\Alpha$, it holds that $\Alpha_i\neq \Alpha_i',\forall i\in T_1\cap T_2$. Now, further conditioned on $\Alpha$, it holds by Lemma~\ref{lemma:avgH-minH} that $\avgH(X\mid S=s,\Alpha=\alpha)\ge k_3+\log(1/\delta)$, $\avgH(X'\mid S=s,\Alpha=\alpha)\ge k_3+\log(1/\delta)$; $\forall i\in T_1$, with probability $1-\sqrt{\eps_1}$, $\avgH(R_i \mid \Alpha=\alpha) \ge d_1-ta$, $\forall i\in T_2$, with probability $1-\sqrt{\eps_1}$, $\avgH(R'_i \mid \Alpha=\alpha) \ge d_1-ta$. Also note that conditioned on $\Alpha$, $X$ ($X'$) is independent of $\cbra{R_i}_{i\in T_1}$ and $\cbra{R'_i}_{i\in T_2}$ --- since if we fix in order $X_1$, and then $Y_1$ (say, a slice from $R_i$), and then $W_1$ and $W_2$ in Algorithm~\ref{alg:advGen}, the random variables $X_1$, and then $Y_1$, and then $W_1$ and $W_2$ are deterministic function of $X$, $R_i$, $Y$ (since after $X_1$ is fixed, $W_1$ is a linear function of $Y$), and $X$ (since after $Y_1$ is fixed, $W_2$ is a linear function of $X$) respectively.
    
    Conditioned on $\Alpha_i\neq \Alpha_i',\forall i\in T_1\cap T_2$, by Theorem~\ref{thm:AffineAdvCB}, Lemma~\ref{lemma:sext-deficient}, Lemma~\ref{lemma:acb-to-U_t} and Lemma~\ref{lemma:minH-avgH}, it holds that $Z_{T_1,T_2}\approx_{O(t(\frac{2^{at}}{\sqrt{\eps_1}} (t\eps_3+\delta)+\sqrt{\eps_1}+\sqrt{\eps_1}))}U_t \mid (Y,Y')$.
    
    Set $\gamma = O(t(\frac{2^{at}}{\sqrt{\eps_1}} (t\eps_3+\delta)+\sqrt{\eps_1}+\sqrt{\eps_1}))$ be the statistical distance between $Z_{T_1,T_2}$ and $U_t$, it holds that $\Pr[\BAD_{T_1,T_2} \mid S=s]\le t(\sqrt{\eps_1}+\eps_2+\delta)$, and $\Pr[\BAD_{T_1,T_2}]\le \Pr[A_{T_1,T_2}]t(\sqrt{\eps_1}+\eps_2+\delta)\le t(\sqrt{\eps_1}+\eps_2+\delta)$. Apply a union bound on all choices of $(T_1,T_2)$, we have \[\Pr[\BAD_x]\le \binom{2D}{t}t(\sqrt{\eps_1}+\eps_2+\delta) \le (2D)^tt(\sqrt{\eps_1}+\eps_2+\delta).\]
    It follows that with probability $1-2^{-k'+1}-(2D)^tt(\sqrt{\eps_1}+\eps_2+\delta)=1-n^{-\Omega(1)}$, it holds that $(Z_1,\cdots,Z_D,Z_1',\cdots,Z_D')$ is a $(2(\sqrt{\eps_1}+\eps')D,t,\gamma)$-NOBF source on $\bin^{2D}$. 
    \end{proof}
    Since it holds that $t\ge c_{\ref{thm:ac0-xor}}\log^{21}(D)$, $(2D)^{t+1}\gamma <1$, and $2(\sqrt{\eps_1}+\eps')D=o(D)$, by Theorem~\ref{thm:ac0-xor}, 
    \[
    W\oplus W'\approx_{n^{-\Omega(1)}} U_m \mid (Y,Y').
    \]
    By Lemma~\ref{lemma:nm-xor},
    \[
        W\approx_{n^{-\Omega(1)}} U_m \mid (Y,Y',W'),
    \]
    Since $W$ is close to uniform independent of $Y$ and $(Y',W')$, moreover, conditioned on $V'$, $\avgH(Y)=\Omega(k)$, by properties of strong seeded extractors, it holds that 
    \[
       V\approx_{n^{-\Omega(1)}} U_{m'} \mid V'
    \]
    concluding that $\tnmext$ is a two-source non-malleable extractor with error $n^{-\Omega(1)}$.
\end{description}
\end{proof}

\subsection{Affine Non-Malleable Extractors}

\begin{algorithm}[H]
    \caption{$\anmExt(x)$ for $O(\log^C(n))$ entropy}
    \label{alg:anmExt-pe}
    \begin{algorithmic}
        \medskip
        \State \textbf{Input:} $x \in \bin^n$ --- an $n$ bit string, an entropy threshold $k=O(\log^C(n))$.
        \State \textbf{Output:} $w \in \bin^{\Omega(k)}$.
        \State \textbf{Subroutines and Parameters:} \\
        Let $\LExt:\bin^n\times \bin^d \to \bin^{n_1}$ be the strong linear seeded extractor from Theorem~\ref{thm:LExt-optimal-seedlength} with seed length $d = c_1\log(n)$ for some constant $c_1$ to extract from entropy $k_1 = n_1^2$. The output length $n_1$ to be decided later and the error $\eps_1 = n^{-2}$. \\
        Let $\AffineAdvGen: \bin^n \times \bin^{n_1} \to \bin^a$ be the advice generator from Theorem~\ref{thm:AffineAdvGen} with advice $a = O(\log(n/\eps_2))$ and error $\eps_2 = \frac{1}{20t^2\cdot(2D)^{t+1}}$. \\ 
        Let $\AffineAdvCB:\bin^n \times \bin^{n_1} \times \bin^{d+a} \to \bin$ be a $t$-affine correlation breaker from Theorem~\ref{thm:AffineAdvCB} with entropy $k_2=O(t(a+d+1)+t^2\log(n/\eps_3))$, seed length $n_1=O(t(a+d+1)+t\log^3(t+1)\log(n/\eps_3))$ where $\eps_3 = O(\frac{1}{2^{at}\cdot (20t^2)\cdot (2D)^{t+1}})$.\\
        Let $\BFExt:\bin^D \to \bin^m$ be the extractor from Theorem~\ref{thm:ac0-xor} with $\delta = 1/c_1$ and error $\eps_4 = n^{-\Omega(1)}$. 
        \\\hrulefill \\
        \begin{enumerate}
            \item Let $y_{i} = \LExt(x,i)$ for every $i\in [D]$.
            \item Let $r_i = \AffineAdvCB(x, y_i,\AffineAdvGen(x,y_i)\circ i)$ for every $i\in [D]$.
            \item Output $w=\BFExt(r_1,\cdots,r_D)$.
        \end{enumerate} 
    \end{algorithmic}
\end{algorithm}

\begin{theorem}
    There exists an explicit construction of an affine non-malleable extractor for entropy $k\ge \polylog(n)$ with output length $\Omega(k)$ and error $n^{\Omega(1)}$.
\end{theorem}
\begin{proof}
Let $Y = Y_1 \circ \cdots \circ Y_D$ and $Y'=Y_1' \circ \cdots \circ Y_D'$ be the concatenation of outputs of $\LExt$ on the input $X$ and the tampered counterpart $X'$.

 \begin{description}
     \item [Case 1.] $H(X')\ge \max\cbra{k_1,k_2+at}$. \\
    In this case, since $H(X)\ge k_1$, $H(X')\ge k_1$, by Lemma~\ref{lemma:affine-uniform}, $\exists I\subset [D],\;|I| \ge (1-2\eps_1)D$ s.t. $Y_i= U_{n_1}$; $\exists I'\subset [D],\;|I'| \ge (1-2\eps_1)D$ s.t. $Y'_i= U_{n_1}$. 
    
    Now we define the set of ``good" rows $G:=\{Y_i:i\in I\} \sqcup \{Y_i':i\in I'\}$ 
    
    In total, there are at least $(1-2\eps_1)\cdot 2D$ ``good" inputs to $\AffineAdvCB$, i.e., $|G|\ge (1-2\eps_1)\cdot 2D$. 
    
    To show $t$-wise independence of the $``R"$ random variables generated from random variables in $G$, we need to exploit the property of $\AffineAdvCB$, i.e., Eqn.~(\ref{eqn:acb}), and eventually we will apply Lemma~\ref{lemma:acb-to-U_t}. 
    
    To achieve this, for each instance of $\AffineAdvCB$ in Eqn.~(\ref{eqn:acb}), we need each $Y_i$ ($Y'_j$) to be uniform and its advice differs from others. This guarantees for any incurrence of Eqn.~(\ref{eqn:acb}), each of the conditioned bits is close to uniform conditioned on other bits including the one that is inferred to be close to uniform.
    
    For any $Y_{i:i\in I\setminus I'}$ or $Y'_{i:i\in I'\setminus I}$ from the good rows, this is automatically guaranteed. For rows from $\{Y_i,Y_i':i\in(I'\cap I)\}$, we need to condition on the event that the advice or the uniform seed is different from the advice of the counterpart. Conditioning on the advice strings can cause entropy loss to $X$, $X'$, $R$ and $R'$. The worst scenario is where we have $t/2$ rows from $\{Y_i:i\in(I_1'\cap I_1 )\}$ and their counterparts from $\{Y_i':i\in(I'\cap I)\}$ so that we in total need to fix $t$ advice strings. We detail the analysis in the following: imagine having $t/2$ rows from $R$ that are indexed from $(I\cap I')$. Conditioned on $\alpha_{i_1},\alpha'_{i_1},\alpha_{i_2},\alpha'_{i_2},\cdots,\alpha_{i_{t/2}},\alpha'_{i_{t/2}}$, it holds that
    \begin{itemize}
        \item $H(X)\ge k_2+at -at\ge k_2$, $H(X')\ge k_2+at -at \ge k_2$
        \item $H(Y_{i_j})\ge n_1-at,\;H(Y'_{i_j})\ge n_1-at\;\forall j\in [t/2]$
    \end{itemize}
    With probability $1-t\eps_2$, it holds that 
    \[
        \alpha_{i_j} \neq \alpha'_{i_j},\;\forall j\in [t/2]
    \]
    Note that here we allow the seed to have some entropy deficiency, i.e., $Y_i,Y_i'$ be $(n_1,n_1-at)$-source, and bound the final error by Lemma~\ref{lemma:sext-deficient}.
    
    Let $\mathcal R:=\{R_{i_j},R'_{i_j}\}_{j\in[t/2]}$. By Definition~\ref{def:AffineAdvCB} and Lemma~\ref{lemma:sext-deficient}, for an arbitrary $R\in \mathcal R$ we have
    \[
    (R\approx_{2^{at}\eps_3+t\eps_2} U)\mid \mathcal{R}\setminus\{R\}.
    \] 
    Therefore, it holds by Lemma~\ref{lemma:acb-to-U_t}:
    \[
        (R_{i_j},R'_{i_j})_{j\in[t/2]} \approx_{t(2^{at}O(t\eps_3)+t\eps_2)} U_t.
    \]
    Therefore,  
    we have $(R_1,\cdots, R_D,R_1',\cdots,R_D')$ is a $(4\eps_1D,t,t(t\eps_2+2^{at}O(t\eps_3)))$-NOBF source. 
    
    Since it holds that $t\ge c_{\ref{thm:ac0-xor}}\log^{21}(D)$, $(2D)^{t+1}\cdot t(t\eps_2+2^{at}O(t\eps_3)) <1$, and $n^{-\Omega(1)}+\frac{4\eps_1D}{D^{1-1/c_1}}=n^{\Omega(1)}+\frac{4D^{1-2/c_1}}{D^{1-1/c_1}}=n^{-\Omega(1)}$
    by Theorem~\ref{thm:ac0-xor}, 
    \[
    W\oplus W'\approx_{n^{-\Omega(1)}} U_m.
    \]
    By Lemma~\ref{lemma:nm-xor},
    \[
        W\approx_{n^{-\Omega(1)}} U_m \mid W',
    \]
    now we condition on $W$ and $W'$
    concluding that $\anmExt$ is an affine non-malleable extractor with error $n^{-\Omega(1)}$.
     \item [Case 2.] $H(X \mid X') \ge \max\cbra{k_1,k_2+at}$. \\
 In this case, we can fix $X'$ and note that $\anmExt(X')$ is also fixed. Since $H(X \mid X') \ge \max\cbra{k_1,k_2+at}$, it is clear that $\anmExt(X \mid X')\approx_{n^{-\Omega(1)}} U_m$.
 \end{description}
\end{proof}

\section{Non-Malleable Extractors for \texorpdfstring{$O(\log n)$}{Lg} entropy}
\label{sec:nmext-const}

In this section, we present two-source and affine non-malleable extractors that works for $O(\log n)$ entropy.

\subsection{two-source Non-Malleable Extractors}

We adapt the scheme in~\cite{BADTS:STOC:17} with new techniques alluded to in Technical Overview~\ref{subsec:nobf}.

\begin{algorithm}[H]
    \caption{$\tnmext(x,y)$ for $O(\log n)$ entropy}
    \label{alg:tmmExt-ce}
    \begin{algorithmic}
        \medskip
        \State \textbf{Input:} $x,y \in \bin^n$ --- two $n$ bit strings, $\eps\in (0,1)$ a constant.
        \State \textbf{Output:} $w \in \bin$.
        \State \textbf{Subroutines and Parameters:} \\
        Let $\delta'=1/10$, $\alpha=1/2$. \\
        Let $\srSamp(x,s,z):=\Ext(x,\Gamma(s,z))$ with the $\Ext$ and $\Gamma$ set up below is a $(\eps_1,\delta_1)$-somewhere random sampler for entropy $k_1+\log(1/\delta_1)$ where $\delta_1=\eps/32$ and $k_1$ will be decided later. \\
        \begin{itemize}
            \item Let $\Ext:\bin^n \times \bin^{m_1} \to \bin^{m_2}$ from Theorem~\ref{thm:GUV09} with $m_1=c_{\ref{thm:GUV09}}\log(n/\eps_0)$, output length $m_2$ to be decided later, entropy $k_1=2m_2$ and
            error $\eps_0=1/n$. 
            \item Let $\Gamma:[D] \times [B] \to [M_1=2^{m_1}]$ be a $(K_{\Gamma}=(1/\eps_1)^{\alpha},\eps_{\Gamma}=3\eps_0)$-disperser from Lemma~\ref{lemma:disperser-Zuc07} where $D=(1/\eps_1)^{1+\alpha}$ and $B=O\pbra{\frac{(1+\alpha)\log(1/\eps)}{\log(1/(3\eps_0))}}=O\pbra{\frac{\log(1/\eps_1)}{\frac{\alpha}{4c_{\ref{thm:GUV09}}}\log(1/\eps_1)-1}}=O(1)$. Since $M_1=\sqrt{K_\Gamma}$, we have $\eps_1 = \eps_0^{4c_{\ref{thm:GUV09}}/\alpha} = (1/n)^{4c_{\ref{thm:GUV09}}/\alpha}$. 
        \end{itemize}\\
        Note that 
        \begin{align}
            2\eps_1=2D^{-1/(1+\alpha)}\le (2D)^{-1/2-\delta'}.
        \end{align}
        Let $\Ext':\bin^n \times \bin^{m_2} \to \bin^{d'}$ be a $(k',\eps')$-strong seeded extractor from Theorem~\ref{thm:GUV09} with entropy $k'=O(d')$, seed length $m_2=O(\log(n/\eps'))$, output length $d'$ to be decided later, and error $\eps'=O\pbra{\frac{\eps^2}{(2D)^{2t}t^2B^2}}$. \\
        Let $\advGen: \bin^n \times \bin^n \times \bin^{d_1} \to \bin^a$ be the advice generator from Theorem~\ref{thm:advGen} with entropy $k_2=O(\log(1/\eps_2))$, $d_1=O(\log(n/\eps_2))$, advice length $a=O(\log(1/\eps_2))$ and error $\eps_2=\frac{\eps}{(2D)^ttB}$.\\
        Let $\AffineAdvCB:\bin^n \times \bin^{d'} \times \bin^{a+d+b} \to \bin$ be a $tB$-affine advice correlation breaker from Theorem~\ref{thm:AffineAdvCB} with entropy $k_3=(t(a+d+b+1)+(tB)^2\log(n/\eps_3))$, seed length $d'=c_{\ref{thm:AffineAdvCB}}(t(a+d+b+1)+(tB)\log^3(tB+1)\log(n/\eps_3))$ and error $\eps_3=O\pbra{\frac{\eps^3}{(2D)^{3t}t^42^{2atB}}}$.

        \\
        Let $\Maj:\bin^D \to \bin$ be the Majority function from Lemma~\ref{lemma:maj}. 
        \\\hrulefill \\
        \begin{enumerate}
            \item Let $x_{s,z} = \srSamp(x,s,z)$.
            \item Let $y_{s,z} = \Ext'(y,x_{s,z})$. 
            \item Let $\alpha_{s,z} =\advGen(x,y,\Slice(y_{s,z},d_1))$. 
            \item Let $r_{s,z} = \AffineAdvCB(x,y_{s,z}, \alpha_{s,z}\circ (s,z))$.
            \item Let $z_i = \oplus^B_{j=1} r_{i,j}$.
            \item Output $w=\Maj(z_1,\cdots,z_D)$.
        \end{enumerate} 
    \end{algorithmic}
\end{algorithm}

\begin{theorem}\label{thm:2nmext-ce}
    For every constant $\eps>0$ there exists a constant $c>1$ and a family of functions $\tnmext:\bin^n\to \bin$, such that for any two independent sources $X$ and $Y$ of min-entropy at least $k\ge c\log n$, any tampering functions $f,g$ such that at least one of $f$ and $g$ has no fixed points, it holds that 
    \[
        (\tnmext(X,Y),\tnmext(f(X),g(Y)))\approx_\eps (U_1,\tnmext(f(X),g(Y))).
    \]
\end{theorem}
\begin{proof}
Let $X'=f(X)$, $Y'=g(Y)$. Let $\eps = O(1)$ be a given constant, and $\delta = O\pbra{\frac{\eps^2}{(2D)^{2t}t^22^{atB}}}$ be a error parameter that will be used in the analysis below. Let $k_x = \max\{k_2+tBm_2+\log(1/\delta),k_3+tBm_2+atB+\log(1/\delta),k_1+\log(1/\delta_1)\}$ and $k_y=k'$. 

Let $\minH(X)\ge 2k_x+\log(2/\eps)$ and $\minH(Y) \ge 2k_y+ \log(2/\eps)$.
\begin{description}
    \item [Case 1.] $\minH(Y') \le k_y$. 
    
    Without loss of generality, assume that $Y$ is a flat source. Let $\BAD$ be defined as $\BAD:=\{y'\in g(Y):|g^{-1}(y')|\le 2^{k_y}\}$. 
    \begin{itemize}
        \item If $\Pr[y'\in \BAD]\le \eps/2$, our analysis can be confined to the subsources \(Y\mid \overline{\BAD}\). Note that we use the notation \(\overline{\BAD}\) to represent the event \(y'\in \Supp(Y)\setminus \BAD\). Given that these subsources contribute to a density of \(1-\eps/2\), neglecting \(Y\mid \BAD\) results in only an \(\eps/2\) deviation in the final error.
        \item Otherwise, \(\Pr[y'\in\BAD] > \eps/2\), then we can infer that 
        \[\max_{y'\in\BAD}\Pr[Y'=y'\mid Y'\in \BAD]\le \frac{2^{k_y}\cdot 2}{\eps|\Supp(Y)|}.\] From this, we deduce the following inequalities: 
        \(\minH(Y'\mid Y'\in\BAD) \ge \minH(Y)-k_y-\log(2/\eps)\) and \(\minH(Y \mid Y'\in \BAD)\ge -\log\pbra{\frac{2}{\eps|\Supp(Y)|}} \ge \minH(Y)-\log(2/\eps)\). We can then think of $(Y \mid Y'\in\BAD)$ and $(Y' \mid Y'\in \BAD)$ as a new pair of source and its tampering by the function $g$ restricted to $(Y\mid Y'\in \BAD)$ such that both $(Y\mid Y'\in \BAD)$ and $(Y' \mid Y'\in \BAD)$ meets the minimum entropy requirement for $\tnmext$ for the second source. And the rest of the analysis only depends on the entropy of the first source and its tampering by $f$. However, in any case, we will be able to show that $\tnmext$ runs correctly on the subsources $(Y\mid Y'\in \BAD)$ and $(Y'\mid Y'\in \BAD)$.
    \end{itemize}
    Now consider $y'\in\overline{\BAD}$. It holds that $\minH(Y\mid Y'=y')\ge k_y$, meaning that the entropy of $Y$ conditioned on $(Y'=y')$ still meets the requirement to carry out all the functionality of the extractors in the rest of steps, we can safely condition $Y$ on $Y'$. Now we additionally fix $W'=\tnmext(X',Y')$. From this point, we could just focus on the running of $\tnmext$ on $(X\mid W)$, $(Y \mid Y')$. And by the above bounds on the average-case entropy of the above two sub-sources, since $\tnmext$ itself resembles an algorithm for $\mathsf{2Ext}$, the correctness follows. Now, since $Y$ is just a convex combination of subsources $(Y\mid Y'\in \BAD)$ and $\{(Y\mid Y'=y')\}_{y'\in \overline{\BAD}}$, it holds that $\tnmext$ runs correctly on $(X,Y)$ and $(X',Y')$.\\
    From now on, we assume that $\minH(Y')>k_y$. We need to consider two cases based on $\minH(X')$.
    \item [Case 2.] $\minH(X')\le k_x$.

    Without loss of generality, assume that \(X\) is a flat source. Let \(\BAD\) be defined as \(\BAD:=\{x'\in f(X):\abs{f^{-1}(x')} \le 2^{k_x}\}\). 
    \begin{itemize}
        \item If \(\Pr[x'\in\BAD] \le \eps/2\) our analysis can be confined to the subsources \(X\mid \overline{\BAD}\). Note that we use the notation \(\overline{\BAD}\) to represent the event \(x'\in \Supp(X)\setminus \BAD\). Given that these subsources contribute to a density of \(1-\eps/2\), neglecting \(X\mid \BAD\) results in only an \(\eps/2\) deviation in the final error.
        \item Otherwise, \(\Pr[x'\in\BAD] > \eps/2\), then we can infer that \[\max_{x\in\BAD}\Pr[X'=x'\mid X'\in \BAD]\le \frac{2^{k_x}\cdot 2}{\eps|\Supp(X)|}.\] From this, we deduce the following inequalities: \(\minH(X '\mid X'\in\BAD) \ge \minH(X)-k_x-\log(2/\eps)\) and \(\minH(X \mid X'\in\BAD)\ge -\log\pbra{\frac{2}{\eps|\Supp(X)|}} \ge \minH(X)-\log(2/\eps)\). Then by what we will show in Case 3 in the following paragraphs, it holds that 
        \[
            (\tnmext(X\mid \BAD,Y),\tnmext(\BAD,g(Y)))\approx_{\eps/2}(U_m,\tnmext(\BAD,g(Y)))
        \]
    \end{itemize}

    Now consider $x'\in\overline{\BAD}$. It holds that $\minH(X\mid X'=x')\ge k_x$, meaning that the entropy of $X$ conditioned on $(X'=x')$ still meets the requirement to carry out all the functionality of the extractors in the rest of steps, we can safely condition $X$ on $X'$. Now we additionally fix $\tnmext(X',Y')$, and the rest of the analysis is similar to the case where $\minH(Y')\le k_y$. 
    \item [Case 3.] $\minH(X')\ge k_x$, $\minH(Y')\ge k_y$. 

    \begin{lemma}\label{lemma:use-x-samp-y-const}
        With probability $1-2\delta_1$ over the fixing of $X$, for no less than $(1-2\eps_1)D$ number of $s\in [D]$, it holds that $\exists z\in [B]$ s.t. $\Ext'(Y,x_{s,z})\approx_{\sqrt{\eps'}}U_{d'}$; for no less than $(1-2\eps_1)D$ number of $s\in [D]$, it holds that $\exists z\in [B]$ s.t. $\Ext'(Y',x_{s,z})\approx_{\sqrt{\eps'}}U_{d'}$.
    \end{lemma}
    \begin{proof}
        Let $V$ be a uniform distribution on $\bin^{m_2}$, since $\minH(Y)\ge k'$, it holds by Theorem~\ref{thm:GUV09} that $|(\Ext'(Y,V),V)-(U_{d'},V)|\le \eps'$. Let $\BAD_1\subseteq [D]$ s.t. $\forall i\in \BAD_1$, it holds that $|(\Ext'(Y,i),i)-(U_{d'},i)|\ge \sqrt{\eps'}$. Then it holds that $|\BAD_1|\le \sqrt{\eps'}D$ since $\frac{1}{D}\sum_{i\in \BAD_1}|(\Ext'(Y,i),i)-(U_{d'},i)|\le |(\Ext'(Y,V),V)-(U_{d'},V)|\le \eps'$. Therefore, $\forall i\in [D]\setminus \BAD_1$, it holds that $\Ext'(Y,i)\approx_{\sqrt{\eps'}}U_{d'}$; similarly, since $\minH(Y')\ge k'$, $\exists \BAD_2\subset [D],\;|\BAD_2|\le \sqrt{\eps'}D$ s.t. $\forall i\in [D]\setminus \BAD_2$, it holds that $\Ext'(Y',i)\approx_{\sqrt{\eps'}}U_{d'}$.
        
        Now note that by Definition~\ref{def:srs}, since $\sqrt{\eps'}\le \eps_1$, with probability $1-\delta_1$ over the fixings of $X$, no less than $(1-2\eps_1)D$ $s$'s in $[D]$ are good, i.e. $\exists z\in [B]$ s.t. $Y_{s,z}\approx_{\sqrt{\eps'}}U_{d'}$; with probability $1-\delta_1$ over the fixings of $X'$ (thus also $X$, since $X'$ is a deterministic function of $X$), no less than $(1-2\eps_1)D$ $s$'s in $[D]$ are good, i.e. $Y'_{s,z}\approx_{\sqrt{\eps'}}U_{d'}$. By a union bound, the statement of Lemma~\ref{lemma:use-x-samp-y-const} holds.
    \end{proof}
    For all $s\in [D]$, let $Y_s = Y_{s,1} \circ Y_{s,2}\circ \cdots \circ Y_{s,B}$, $Y_s' = Y'_{s,1}\circ Y'_{s,2}\circ \cdots \circ Y'_{s,B}$. Let $\widetilde{Y}$ denote the matrix whose $s$-th row is $Y_s$ and $\widetilde{Y}'$ the matrix whose $s$-th row is $Y_s'$.
    \begin{lemma}
        Conditioned on the event that no less than $(1-2\eps_1)D$ rows in $\widetilde{Y}$ and no less than $(1-2\eps_1)D$ rows in $\widetilde{Y}'$ are close to uniform. It holds that with probability at least $1-\eps/8$ over the fixing of $X$, $(Z_1,\cdots,Z_D,Z'_1,\cdots,Z'_D)$ is a $(4\eps_1D,t,O(t(\frac{2^{atB}}{\sqrt{\eps'}}(t\eps_3+\delta)+2\sqrt{\eps'})))\text{-}\NOBF$ source.
    \end{lemma}
    \begin{proof}
        Define a ``bad" $x\in \Supp(X)$ to be a string that satisfies the following two properties:
        \begin{enumerate}
            \item [(a)] Conditioned on the fixing of $X=x$, there exists a large set $G_1\subset [D],|G_1|\ge (1-2\eps_1)D$; $G_2\subseteq [D],|G_2|\ge(1-2\eps_1)D$ s.t. $\forall s\in G_1$, $Y_s$ is $\sqrt{\eps'}$-close to somewhere random; $\forall s\in G_2$, $Y'_s$ is $\sqrt{\eps'}$-close to somewhere random.
            \item [(b)] Conditioned on the fixing of $X=x$, there exists two subsets $T_1\subset G_1,T_2\subset G_2$ with $|T_1|+|T_2|=t$ such that the concatenation of bits in $Z$ with indices in $T_1$ and bits in $Z'$ with indices in $T_2$, which we denoted by $Z_{T_1,T_2}$, is $\gamma$-far from uniform, where $\gamma$ is an error parameter to be decided later.
        \end{enumerate}
        For each $T_1,T_2\subset [D]$ with $|T_1|+|T_2|=t$, we now define an event $\BAD_{T_1,T_2}$ to be the set of $x$'s in $\Supp(X)$ that satisfies the following two properties:
        \begin{enumerate}
            \item[(c)] Conditioned on the fixing of $X=x$, each of $Y_s$, $s\in G_1$ and $Y'_s$, $s\in G_2$, with index $s\in G_2$ is $\sqrt{\eps'}$-close to somewhere random.
            \item[(d)] Conditioned on the fixing of $X=x$, $Z_{T_1,T_2}$ is $\gamma$-far from uniform.
        \end{enumerate}
        Observe that the property (c) is determined by $tB$ random variables $\{X_{s,z}\}_{s\in T_1,z\in[B]}$ and $\{X'_{s,z}\}_{s\in T_2,z\in [B]}$. Let $X_{T_1}$ be the concatenation of $\{X_{s,z}\}_{s\in T_1,z\in[B]}$. Let $X_{T_2}$ be the concatenation of $\{X'_{s,z}\}_{s\in T_2,z\in [B]}$. Let $X_{T_1,T_2}=X_{T_1}\circ X_{T_2}$. Define $A_{T_1,T_2}$ to be the set of $x_{T_1,T_2}$'s in $\Supp(X_T)$ that makes property (c) satisfied, then we have $\Pr[\BAD_{T_1,T_2}]=\sum_{x_{T_1,T_2}\in X_{T_1,T_2}} \Pr[X_{T_1,T_2}=x_{T_1,T_2}]\Pr[\BAD_{T_1,T_2} \mid X_{T_1,T_2}=x_{T_1,T_2}]$.

        Now since the size of $X_{T_1,T_2}$ is small: $tBm_2$, $\avgH(X \mid X_{T_1,T_2}=x_{T_1,T_2})\ge k_2+\log(1/\delta)$, $\avgH(X'\mid X_{T_1,T_2}=x_{T_1,T_2})\ge k_2+\log(1/\delta)$.

        Define the random variables $\Alpha_{s,z}=\advGen(X,Y,Y_{s,z})$; $\Alpha_{s,z}'=\advGen(X',Y',Y_{s,z}')$. Let $\Alpha$ be the concatenation of $\Alpha_{s,z}$-r.v.s for all $(s,z)\in T_1\times[B]$ and $\Alpha'_{s,z}$-r.v.s for all $(s,z)\in T_2\times[B]$. 
        With probability no less than $1-tB(\eps_2+\sqrt{\eps'}+\delta)$ over the fixing of $\Alpha$, it holds that $\Alpha_{s,z}\neq \Alpha_{s',z'}',\forall (s,z)\in T_1\times [B]$ and $\forall (s',z')\in T_2\times [B]$. Now, further conditioned on $\Alpha$, it holds by Lemma~\ref{lemma:avgH-minH} that $\avgH(X\mid X_{T_1,T_2}=x_{T_1,T_2},\Alpha=\alpha)\ge k_3+\log(1/\delta)$, $\avgH(X'\mid X_{T_1,T_2}=x_{T_1,T_2},\Alpha=\alpha)\ge k_3+\log(1/\delta)$; $\forall s\in T_1$, $\exists z\in [D]$ s.t. with probability $1-\sqrt{\eps'}$, $\avgH(Y_{s,z} \mid \Alpha=\alpha) \ge d_1-tBa$; $\forall s\in T_2$, $\exists z\in [D]$, with probability $1-\sqrt{\eps'}$, $\avgH(Y'_{s,z} \mid \Alpha=\alpha) \ge d_1-tBa$. Also note that conditioned on $\Alpha$, $X$ ($X'$) is independent of $\cbra{Y_{s,z}}_{(s,z)\in T_1 \times [B]}$ and $\cbra{Y'_{s,z}}_{(s,z)\in T_2 \times [B]}$ for similar reason as stated in the proof of Theorem~\ref{thm:2nmext-pe}.

        Conditioned on $\Alpha_{s,z}\neq \Alpha_{s',z'}',\forall (s,z)\in T_1\times [B]$ and $\forall (s',z')\in T_2\times [B]$, by Theorem~\ref{thm:AffineAdvCB}, Lemma~\ref{lemma:minH-avgH}, Lemma~\ref{lemma:sext-deficient} and Lemma~\ref{lemma:acb-to-U_t}, it holds that $Z_{T_1,T_2}\approx_{O(t(\frac{2^{atB}}{\sqrt{\eps'}}(t\eps_3+\delta)+\sqrt{\eps'}+\sqrt{\eps'}))}U_t$. 

        Set $\gamma = O(t(\frac{2^{atB}}{\sqrt{\eps'}}(t\eps_3+\delta)+2\sqrt{\eps'}))$ to be the statistical distance between $Z_{T_1,T_2}$ and $U_t$, it holds that $\Pr[\BAD_{T_1,T_2} \mid X_{T_1,T_2}=x_{T_1,T_2}]\le tB(\sqrt{\eps'}+\eps_2+\delta)$, and $\Pr[\BAD_{T_1,T_2}]\le \Pr[A_{T_1,T_2}]tB(\sqrt{\eps'}+\eps_2+\delta)\le tB(\sqrt{\eps'}+\eps_2+\delta)$. Apply a union bound on all choices of $T_1$ and $T_2$, we have
        \[\Pr[\BAD_x]\le  \binom{2D}{t}tB(\sqrt{\eps'}+\eps_2+\delta) \le (2D)^{t}tB(\sqrt{\eps'}+\eps_2+\delta).\]
        It follows that with probability $1-2\delta_1-(2D)^ttB(\eps_2+\sqrt{\eps'}+\delta)\ge 1-\eps/8$, it holds that $(Z_1,\cdots,Z_D,Z_1',\cdots,Z_D')$ is a $(4\eps_1D,t,\gamma)$-NOBF source on $\bin^{2D}$. 
    \end{proof}
    Set $t$ and $n$ large enough so that $C_{\ref{thm:maj-xor}}\pbra{\frac{1}{\sqrt{t}}}\le \eps/24$ and $C_{\ref{thm:maj-xor}}(2D)^{-\delta'}\le \eps/24$, then it holds that  
    \[
        C_{\ref{thm:maj-xor}}\pbra{\frac{1}{\sqrt{t}}+(2D)^{-\delta'}}+(2D)^t \cdot\gamma\le \frac{\eps}{12}+\frac{\eps}{24} = \eps/8
    \]
    By Theorem~\ref{thm:maj-xor}, it holds that 
    \[
        W\oplus W' \approx_{\eps/8} U_1.
    \]
    By Lemma~\ref{lemma:nm-xor},
    \[
        W\approx_{\eps/2} U_1 \mid W'.
    \]
\end{description}
\end{proof}

\subsection{Affine Non-Malleable Extractors}

Given that the state-of-the-art affine correlation breaker~\cite{Li:FOCS:23} is improved to work over logarithmic seed length, we can modify the parameters of strong linear somewhere random extractor in~\cite{ChattopadhyayGL21} to only extract from logarithmic entropy (their strong linear somewhere random extractor is capable of functioning in such parameters but they only stated slightly inferior parameters due to the bottle-neck from affine correlation breaker then), which is linear in the output length.
\begin{theorem}[Strong Linear Somewhere Random Extractor, modified from~\cite{ChattopadhyayGL21}]\label{thm:linear-sre}
    There exists a constant $\beta$ which satisfies the following. For every constant $\gamma >0$, there exists a constant $C$ such that for every $n\in \N$ and every $c = c(n) < 2^{\sqrt[3]{\log n}}$, there exists an explicit function $\LSRExt:\bin^n \times [D] \times [B] \to \bin^{m}$ such that
    \begin{itemize}
        \item $D\le n^{C}$.
        \item $B\le C\log^2(c(n))$.
        \item For every fixed $s\in [D]$, $z\in [B]$, the function $\LSRExt_{s,z}(x):=\LSRExt(x,s,z)$ is linear.
        \item For every $(n,\beta m)$-affine source $X$, there exists a subset $\BAD\subseteq {D}$ of size at most $D^{\gamma}$ such that for every $s\in[D]\setminus \BAD$, $\exists z\in [B]$ s.t. $\LSRExt(X,s,z)$ \underline{is uniform}.
    \end{itemize}
\end{theorem}

\begin{algorithm}[H]
    \caption{$\anmExt(x)$ for $O(\log n)$ entropy}
    \label{alg:anmExt-ce}
    \begin{algorithmic}
        \medskip
        \State \textbf{Input:} $x\in \bin^n$.
        \State \textbf{Output:} $w \in \bin$.
        \State \textbf{Subroutines and Parameters:} \\
        Let $t = O\pbra{\frac{\log^2(1/\eps)}{\eps^2}}$ be large enough so that $C_{\ref{thm:maj-xor}}\frac{1}{\sqrt{t}}\le \frac{\eps}{12}$.\\
        Let $C_{\ref{thm:linear-sre}}$ be the constant $C$ in Theorem~\ref{thm:linear-sre} with the constant $\gamma=2/5$.\\
        Let $\LSRExt:\bin^n \times [D] \times [B] \to \bin^{m}$ be the strong linear somewhere random extractor from Theorem~\ref{thm:linear-sre} with entropy $k_1 = c_1\log n$, $D\le n^{C_{\ref{thm:linear-sre}}}$, $B\le C_{\ref{thm:linear-sre}}\log^2(c_1)$, and $m$ to be decided later.\\
        Let $\AffineAdvGen: \bin^n \times \bin^{m} \to \bin^a$ be the affine advice generator from Theorem~\ref{thm:AffineAdvGen} with entropy $k_2=O(\log(1/\eps_1))$, advice length $a=O(\log(1/\eps_1))$ and error $\eps_1 = \frac{\eps}{24t^2(2D)^t}$.\\
        Let $\AffineAdvCB:\bin^n \times \bin^{m} \times \bin^{a+d+b} \to \bin$ be a $(2tB)$-affine advice correlation breaker from Theorem~\ref{thm:AffineAdvCB} with entropy $k_3=O(2Bt(a+d+b+1)+(2Bt)^2\log(2Btn/\eps_2))$, seed length $m=O(2Bt(a+d+b+1+\log^3(2Bt+1)\log(2Btn/\eps_2)))$, and error $O(2tB\eps_2)$ where $\eps_2=O\pbra{\frac{\eps}{48 t^2B (2D)^t 2^{at}}}$.\\
        Let $\Maj:\bin^D \to \bin$ be the Majority function from Lemma~\ref{lemma:maj}.
        \\\hrulefill \\
        \begin{enumerate}
            \item Let $y_{s,z} = \LSRExt(x,s,z)$.
            \item Let $r_{s,z} = \AffineAdvCB(x,y_{s,z}, \AffineAdvGen(x,y_{s,z})\circ (s,z))$.
            \item Let $z_i = \oplus^B_{j=1} r_{i,z}$.
            \item Output $w=\Maj(z_1,\cdots,z_D)$.
        \end{enumerate} 
    \end{algorithmic}
\end{algorithm}

\begin{theorem}\label{thm:anmext-ce}
    For every constant $\eps > 0$, there exists a constant $c>1$ such that for every large enough $n$, there exists a family of functions $\anmExt:\bin^n \to \bin$, such that for any affine source $X$ of min-entropy at least $k\ge m$, any affine tampering function $\mathcal{A}$ without fixed points, it holds that 
    \[ (\anmExt(X), \anmExt(\mathcal{A}(X)))  \approx_{\eps} (U_1, \anmExt(\mathcal{A}(X))).\] 
\end{theorem}
\begin{proof}
Define $X':=\mathcal{A}(X)$. Our proof bifurcates depending on $H(X')$. \\

\textbf{Case 1.} $H(X') \ge \max\cbra{k_1,k_3+tBm+ta}$. 
\begin{enumerate}
    \item Since $H(X),H(X') \ge k_1$, by Theorem~\ref{thm:linear-sre} it holds that $\exists \BAD,\BAD' \subset [D]$ where $|\BAD|,|\BAD'| \le D^{\gamma}$ such that $\forall s \in [D]\setminus \BAD$, $\exists z\in [B]$ s.t. $Y_{s,z}= U_{m}$; $\forall s \in D\setminus \BAD'$, $\exists z\in [B]$ s.t. $Y'_{s,z}= U_{m}$.
    \item To prove the $t$-wise independence property, we will only apply the uniform rows of $Y_{s,[B]},\forall s\in [D]\setminus \BAD$ and $Y'_{s,[B]},\forall s\in [D]\setminus \BAD'$ to the first coordinate of $\AffineAdvCB$. Assume that we put $Y_{s,z}$ to the first coordinate and $s\in ([D]\setminus \BAD)\cap ([D]\setminus \BAD')$, then $Y'_{s,z}$ will be among the variables being conditioned on. And we need to guarantee that the advice of $Y_{s,z}$ and $Y'_{s,z}$ differs --- that is if the $s$'th matrices of $Y$ and $Y'$ are both somewhere random, then the counterpart's advice need to be different from the uniform row's. 
    
    Let $T_1\subset [D]\setminus \BAD$, $|T_1|=t_1$, and $T_2 \subset [D]\setminus \BAD'$, $|T_2| = t_2$ s.t. $t_1+t_2 = t$. Now for any $s\in T_1\cap T_2$, let $z_s\in [B]$ s.t. $Y_{s,z_s}=U_m$, $z'_s\in [B]$ s.t. $Y'_{s,z'_s}=U_m$, then to apply $\AffineAdvCB$ using $Y_{s,z_s}$ as the uniform seed and using $Y'_{s,z'_s}$ as the uniform seed, we need to guarantee that $\alpha_{s,z_s}\neq \alpha'_{s,z_s}$ and $\alpha'_{s,z'_s}\neq \alpha_{s,z'_s}$ (note that the advice of other variables that are conditioned are different due to different indices). 
    
    By Theorem~\ref{thm:advGen}, each pair of the above advice are different with probability $1-\eps_1$, and they are all different with probability $\ge 1-t\eps_1$. Let $\overline{\alpha}$ be the concatenation of the advice in $\{\alpha_{s,z_s},\alpha'_{s,z_s},\alpha'_{s,z'_s},\alpha_{s,z'_s}\}_{s\in T_1\cap T_2}$; let $\overline{Y}$ be the concatenation of seeds $\{Y_{s,z}\}_{s\in T_1,z\in [B]}\sqcup \{Y'_{s,z}\}_{s\in T_2,z\in [B]}$. Then by Lemma~\ref{lemma:affine conditioning}, there exists $A,B$ s.t. $A+B = X$, $A$ is independent of $B,\overline{Y}\circ \overline{\alpha}(X)=B$, and $H(A)\ge k-tBm-ta$. Now conditioned on $\overline{\alpha}$, it holds $\forall s\in T_1$ and $\hat{s}\in T_2$  
    \begin{itemize}
        \item $H(Y_{s,z_s})\ge m-at$
        \item $H(Y'_{\hat{s},z'_{\hat{s}}})\ge m -at$
    \end{itemize}
    Let $T\subset \mathcal{R}$ be any subset of size $t$, we need to condition on the advice 
    $\{R_{\bar{s},j}\}_{j\in [B]}\in T$ be any element. Let $j_{\bar{s}}\in [B]$ s.t. $R_{\bar{s},j_{\bar{s}}}=U_{m}$, if
    
    by Definition~\ref{def:AffineAdvCB}
    \begin{align*}
        R^* \approx_{t\eps_1+2^{at} \cdot \eps_2} U \mid \pbra{\bigsqcup_{S'\in (T\setminus\{S\})} S'} \sqcup (S\setminus\{R^*\})
    \end{align*}
    By Lemma~\ref{lemma:acb-to-U_t}, this implies that the set of random variables $\mathcal{Z}:=\{Z_s\}_{s\in [D]\setminus B_1} \sqcup \{Z_s'\}_{s\in [D]\setminus B_2}$ is $(t(t\eps_1+2^{at} \cdot O(tB\eps_2)))$-wise independent.
    \item Therefore, it holds that $(Z_1,\cdots,Z_D,Z'_1,\cdots,Z'_D)$ is $(2D^\gamma,t ,(t(t\eps_1+2^{at} \cdot O(tB\eps_2))))$-NOBF source.
    \item Set $n$ large enough so that $C_{\ref{thm:maj-xor}}(2D)^{\gamma-1/2}\le \eps/12$, then it holds that  
    \[
        C_{\ref{thm:maj-xor}}\pbra{\frac{1}{\sqrt{t}}+(2D)^{\gamma-1/2}}+(2D)^t \cdot(t(t\eps_1+2^{at} \cdot O(tB\eps_2)))\le \frac{\eps}{6}+\frac{\eps}{12} = \eps/4
    \]
    
    By Theorem~\ref{thm:maj-xor}, it holds that 
    \[
        W\oplus W' \approx_{\eps/4} U_1.
    \]
    By Lemma~\ref{lemma:nm-xor},
    \[
        W\approx_{\eps} U_1 \mid W'.
    \]
\end{enumerate}

\textbf{Case 2.} $H(X')<O(\log n)$. \\
In this case, $H(X\mid X') \ge m$, moreover, by Lemma~\ref{lemma:affine conditioning}, conditioned on $X'$, there exists affine sources $A$ and $B$ such that $X=A+B$, $X'$ is a linear function of $B$ and $A$ is independent of $X'$. Since $X'$ is fixed, it holds that $\anmExt(X')$ is also fixed. Now $\anmExt$ coincides with the affine extractor from~\cite{ChattopadhyayGL21} except for that the affine correlation breaker in $\anmExt$ has a longer advice string. Since $H(A) \ge c\log n$, it holds that $\anmExt(X\mid X') \approx_\eps U_1$.

\end{proof}
\section{Conclusion and Open Problems}
\label{sec:open}
In this paper we significantly improved constructions of two-source and affine non-malleable extractors, and our constructions essentially match standard extractors in the regime of small entropy. We note that any future improvement of extractors for NOBF sources (e.g., improvement in the error) can also translate into improvements of our two-source and affine non-malleable extractors. Furthermore, our results suggest that there may be a deeper connection between standard extractors and their non-malleable counterparts, since their constructions and parameters appear quite similar. In particular, previous works have extensively used non-malleable extractors to construct standard extractors, but is it possible that the reverse direction may also be true? That is, can one also use standard extractors to construct non-malleable extractors?
\section*{Acknowledgement}
We thank anonymous reviewers for their helpful comments.

\bibliographystyle{alpha}
\bibliography{nmext}

\end{document}